\documentclass[12pt, draftclsnofoot, onecolumn,romanappendices]{IEEEtran}
\usepackage{bm,cite}
\usepackage{amsmath}
\usepackage{amsthm}
\usepackage{amsbsy}
\usepackage{epsfig}
\usepackage{latexsym}
\usepackage{psfrag}
\usepackage{amssymb} 
\usepackage{placeins}

\newtheorem{theorem}{Theorem}

\newtheorem{proposition}{Proposition}
\newtheorem{lemma}{Lemma}
\newtheorem{corollary}{Corollary}
\DeclareMathAlphabet{\mathbit}{OML}{cmr}{bx}{it}
\DeclareMathAlphabet{\mathsf}{OT1}{cmss}{m}{n}
\DeclareMathAlphabet{\mathTXf}{OT1}{cmss}{bx}{it}

\DeclareMathOperator{\diag}{diag}

\DeclareMathOperator*{\argmin}{argmin}
\DeclareMathOperator*{\argmax}{argmax}

\DeclareMathOperator{\DoF}{DoF}

\newcommand{\I}{\mathbf{I}} 
\newcommand{\norm}[1]{\lVert{#1}\rVert}

\newcommand{\Fro}{{\text{F}}}

\newcommand{\E}{{\mathrm{E}}}

\newcommand{\trans}{{\text{T}}} 
 
\newcommand{\He}{{{\mathrm{H}}}}

\newcommand{\BC}{{\text{BC}}}  

\begin{document} 
\title{Degrees of Freedom of the Network MIMO Channel With Distributed CSI}
\author{Paul de Kerret and David Gesbert\\Mobile Communication Department, Eurecom\\
 2229 route des Cr\^etes, 06560 Sophia Antipolis, France\\\{dekerret,gesbert\}@eurecom.fr}

\maketitle

\begin{abstract}
In this work\footnote{This work has been performed in the framework of the European research project ARTIST4G, which is partly funded by the European Union under its FP7 ICT Objective 1.1 - The Network of the Future.\\ Preliminary results have been published in ISIT 2011, St. Petersburg.}, we discuss the joint precoding with finite rate feedback in the so-called network MIMO where the TXs share the knowledge of the data symbols to be transmitted. We introduce a distributed channel state information (DCSI) model where each TX has its own local estimate of the overall multi-user MIMO channel and must make a precoding decision solely based on the available local CSI. We refer to this channel as the \emph{DCSI-MIMO channel} and the precoding problem as \emph{distributed precoding}. We extend to the DCSI setting the work from Jindal in \cite{Jindal2006} for the conventional MIMO Broadcast Channel (BC) in which the number of Degrees of Freedom (DoFs) achieved by Zero Forcing (ZF) was derived as a function of the scaling in the logarithm of the Signal-to-Noise Ratio (SNR) of the number of quantizing bits. Particularly, we show the seemingly pessimistic result that the number of DoFs \emph{at each user} is limited by the worst CSI across all users and across all TXs. This is in contrast to the conventional MIMO BC where the number of DoFs at one user is solely dependent on the quality of the estimation of his own feedback. Consequently, we provide precoding schemes improving on the achieved number of DoFs. For the two-user case, the derived novel precoder achieves a number of DoFs	limited by the best CSI accuracy across the TXs instead of the worst with conventional ZF. We also advocate the use of {\em hierarchical quantization} of the CSI, for which we show that considerable gains are possible. Finally, we use the previous analysis to derive the DoFs optimal allocation of the feedback bits to the various TXs under a constraint on the size of the aggregate feedback in the network, in the case where conventional ZF is used.
\end{abstract}
\IEEEpeerreviewmaketitle
\section{Introduction}

Network MIMO channel, or multicell MIMO channels, whereby multiple interfering transmitters (TXs) share user messages and allow for joint precoding (downlink), are currently considered for next generation wireless networks \cite{Karakayali2006,Somekh2006,Gesbert2010}. With perfect message and channel state information (CSI) sharing, the different TXs can be seen as a unique virtual multiple-antenna array serving all receivers (RXs), in a multiple-antenna broadcast channel (BC) fashion. 

Although the sharing of user data symbols can be made possible in certain situations, such as cellular networks with a pre-existing backbone infrastructure where user packets can be routed to several base stations simultaneously, the obtaining of accurate CSI at the TXs is made difficult due to the finite quantizing effects over the feedback channels and the limited capability of signaling between TXs to exchange the CSI. In addition, CSI exchange necessarily introduces further degradation due to latency effects over inter-TX links\cite{ARTIST4GD42}. 

This situation gives rise to an interesting information theoretic framework whereby a MIMO broadcast channel is formed (due to the assumed perfect user message sharing among the various TXs), yet the individual TXs composing the distributed multiple-antenna array have access to individual CSI estimates, possibly different from each other, and possibly of different quality (statistically). In this paper, we refer to this channel as the \emph{distributed CSI (DCSI)-MIMO channel}. We emphasize the difference between this CSI model and the previously studied CSI models such as the so-called imperfect limited CSI~\cite{Jindal2006} or the delayed CSI model \cite{Maddah-Ali2010} where the TX antennas are assumed to share ideally the \emph{same} imperfect channel knowledge. 

Note that the sharing of the symbols via finite capacity links between the cooperating TXs has been discussed in recent works\cite{Marsch2008,Simeone2009,Shamai2011,Zakhour2011}. This problem represents in itself a challenging topic, and we consider in the sequel perfect sharing of the users symbols.

For the conventional MIMO BC, the impact of limited feedback \cite{Jindal2006,Caire2007,Au-Yeung2007,Samardzija2006,Ding2007,Yoo2007,Song2008} and the derivation of robust solutions \cite{Shenouda2006,Vucic2009} have been investigated, with later extensions to the multicell coordinated beamforming case \cite{Bjornson2010} and the multicell MIMO case \cite{Kobayashi2009,Tajer2011a,Huh2012}.

More recently, the optimization of the feedback allocations to the different users has been the focus of a large interest. It has been studied in conventional MIMO BCs\cite{Sohn2012}, in multicell settings with coordinated beamforming\cite{Zhang2010,Ozbek2010,Ho2011,Bhagavatula2011a}, in multicell MIMO networks\cite{Saleh2010,Tajer2011b,Khoshnevis2012} and in interfering BCs \cite{Lee2011,Bhagavatula2011b}.

Yet, as mentioned before, these papers always consider \emph{perfect sharing} between the TXs precoding jointly the signal. In contrast, we consider here that each TX has its own imperfect estimation of the multi-user channel but all the TXs jointly precode the user's data symbols. This gives rise to a very different transmission setting which can be seen as a \emph{team decision problem}\cite{Marschak1972}. Indeed, the precoder must cope not only with the inaccuracy of the CSI due to the limited feedback channel capacity but also with the distributiveness of the CSI and the precoding. Each TX emits one component of the transmit signal vector which it computes based on its \emph{own} channel estimate. As is pointed out in this work, the \emph{discrepancies} between the channel estimates obtained by the different TXs are particularly detrimental to the channel capacity, and even to the Degrees of Freedom (DoFs), if not accounted for in the precoding design.

The DCSI-MIMO scenario has meaningful applications to network MIMO schemes in cellular networks or MIMO based multi-TX cooperation in general. It was first studied in \cite{Zakhour2010a}, and a tractable discrete optimization at finite SNR was derived. However, the approach in \cite{Zakhour2010a} does not lend itself to a more general performance analysis, thus giving limited insight for an improved design.

In this paper, we consider the performance of precoding schemes over the DCSI-MIMO channel from a DoFs perspective. The number of DoFs represents the slope with which the rate increases with the SNR in the high SNR regime. Even though it is based on the high SNR analysis, it has been used widely used to gain insight into the wireless transmission thanks to its analytical tractability\cite{Cadambe2008,Maddah-Ali2010}. By essence, the DoFs analysis is not impacted by the unequal pathloss, which can put in question its practical signification in some settings. When all the wireless links present the same pathloss as it is the case in this work, this does not represent an issue. To extend the DoFs analysis to settings with large pathloss differences, it is then more adequate to use the notion of \emph{generalized DoFs} \cite{Etkin2008} which takes the pathloss differences into account. We also assume that our system model is isolated from the rest of the world. In a practical scenario, it follows from the impossibility to serve jointly all the users that there is inevitably interference coming from outside the cooperation area. This implies that the number of DoFs is always zero \cite{Lozano2012a}. The number of DoFs derived inside the cooperation cluster is then representative solely up to an SNR at which point the interference from outside the cooperation area leads to the saturation of the rate.

Our work generalizes to the case of distributed CSI the finite rate feedback study by Jindal~\cite{Jindal2006} for the conventional multiple-antenna BC. In \cite{Jindal2006}, the author derives the number of DoFs as a function of the number of feedback (quantizing) bits exploited by each RX and shows that the number of bits must grow with the logarithm of the SNR in order to preserve the full number of DoFs, using ZF precoding arguments. We also consider ZF schemes as they are known to achieve maximum number of DoFs in wide settings\footnote{Note that the selection of the set of users actually transmitting during one time slot is not considered in this work. In fact the formula for the number of DoFs provided in this work can be used to derive a set of transmitting TX achieving a good number of DoFs, i.e.,  to use a good combination of ZF precoding and time sharing.}. Particularly, a necessary and sufficient feedback of the CSI estimation error for achieving the maximum number of DoFs is derived in \cite{Caire2007} for the compound multiple-antenna BC. This condition is the same as the sufficient condition provided in \cite{Jindal2006}. Thus, no other precoding scheme can achieve the maximal number of DoFs with a lower feedback scaling. This confirms the efficiency of ZF in terms of number of DoFs. As a consequence, we aim in this paper at answering the fundamental questions \emph{"Does conventional ZF also perform well in the distributed MIMO setting?"}, and \emph{"How can we make it more robust in that setting?"}

Specifically, the main contributions read as follows. Let the number of bits quantizing the estimate at TX~$j$ of the normalized channel~$\tilde{\bm{h}}_i^{\He}$ of user~$i$ be $\alpha_i^{(j)}(K-1)\log_2(P)$ with $\alpha_i^{(j)}\in[0,1]$ and $K$ the number of users. Then, we show that in a block fading Rayleigh channel:
\begin{itemize}
\item The number of DoFs achieved at RX~$i$ with conventional ZF is equal to~$\min_{i,j\in\{1,\ldots,K\}}\alpha_i^{(j)}$. Hence, the worst accuracy across all the estimates limits the number of DoFs at each user. This is a pessimistic result and shows a different behavior compared to the conventional MIMO BC.
\item We provide a precoding scheme improving the number of DoFs. In the two-user case, the number of DoFs with the novel precoding scheme is limited by the best accuracy of the CSI across the two TXs instead of the lowest with conventional ZF.
\item To improve the number of DoFs achieved with more users, we introduce a concept of hierarchical quantization of the CSI and we show that this leads to a dramatic improvement of the number of DoFs.
\item Under a total feedback constraint and with ZF schemes, we derive the number of DoFs maximizing allocation of the feedback bits toward each TX.
\end{itemize}
Note that this paper serves to generalize preliminary results that were presented in~\cite{dekerret2011_ISIT}.

\emph{Notations:} We denote by $\Pi_{\mathbf{A}}(\bullet)$ and $\Pi_{\mathbf{A}}^{\perp}(\bullet)$ the orthogonal projectors over the subspace spanned by the matrix~$\mathbf{A}$ and over its orthogonal complement, respectively. $\bar{i}$ denotes the complementary indice of $i$ when only two users are considered, i.e., $\bar{i}=i\mod 2+1$. $\|\bullet\|_{\Fro}$ designates the Frobenius norm while $\mathcal{N}(\mu,\sigma^2)$ denotes the complex circularly symmetric Gaussian distribution with mean~$\mu$ and variance~$\sigma^2$. We also denote the $i$th element of a vector~$\bm{a}$ by $\{\bm{a}\}_i$ and the $(i,j)$th element of a matrix~$\mathbf{A}$ by $\{\mathbf{A}\}_{ij}$. Additionally, we use the notation~$\lesssim$ to denote a relation of order which holds true asymptotically. We also write~$f(x)=o(g(x))$ (resp.~$f(x)=O(g(x))$) to represent the fact that~$\lim_{x\rightarrow \infty} f(x)/g(x)=0$ (resp. $\lim_{x\rightarrow \infty} |f(x)|/|g(x)|\leq a$, with~$a>0$). We also write~$f(x)\sim g(x)$ to denote the fact that~$f(x)=g(x)+o(g(x))$.
\section{System Model}
\subsection{Multicell MIMO}
We consider a joint downlink transmission from $K$~TXs to $K$~RXs using linear precoding and single user decoding. For ease of exposition, the TXs and the RXs are equipped with only one antenna, but the principal elements of our approach could extend in principle to more antennas at the TXs. Similarly, we consider a Rayleigh fading scenario but the approach derived should be valid in many other fading scenarios. The transmission can be described as
\begin{equation}
\begin{bmatrix}
y_1\\
y_2\\
\vdots\\
y_K
\end{bmatrix}
=
\begin{bmatrix}
\bm{h}^{\He}_1\\
\bm{h}^{\He}_2\\
\vdots\\
\bm{h}^{\He}_K
\end{bmatrix}\begin{bmatrix}
x_1\\x_2\\\vdots\\
x_K
\end{bmatrix} 
+
\begin{bmatrix}
\eta_1\\
\eta_2\\
\vdots\\
\eta_K
\end{bmatrix}\\ 
\label{eq:SM_1}
\end{equation}
where $\bm{y}\triangleq[y_1,\ldots,y_K]^{\trans}\in\mathbb{C}^{K\times 1}$ contains the received signals at the RXs, the vector $\bm{x}\triangleq[x_1,\ldots,x_K]^{\trans}\in\mathbb{C}^{K\times 1}$ is defined such that $x_j$ is the signal transmitted by TX~$j$, and $\bm{\eta}\triangleq[\eta_1,\ldots,\eta_K]^{\trans}\in\mathbb{C}^{K\times 1}$ contains the noise realizations at the RXs and has its entries i.i.d. as~$\mathcal{N}(0,1)$.

The vector $\bm{h}^{\He}_i \in \mathbb{C}^{1\times K}$ is the channel from all TXs to the $i$-th RX and define the normalized channel to user~$i$ as $\tilde{\bm{h}}_i\triangleq\bm{h}_i/\norm{\bm{h}_i}$. We also define the multi-user channel matrix~$\mathbf{H}\triangleq [\bm{h}_1,\ldots,\bm{h}_K]^{\He}$ and its normalized counter-part~$\tilde{\mathbf{H}}\triangleq [\tilde{\bm{h}}_1,\ldots,\tilde{\bm{h}}_K]^{\He}$.

The channel is assumed to be block fading and the entries of the channel matrix $\mathbf{H}$ to be i.i.d. as~$\mathcal{N}(0,1)$, modeling a Rayleigh fading channel. The transmitted signal $\bm{x}$ is obtained from the vector of transmit symbols $\bm{s}\triangleq[s_1,\ldots,s_K]^{\trans}\in\mathbb{C}^{K\times 1}$ (whose entries are taken as i.i.d. $\mathcal{N}(0,1)$) as
\begin{equation}
\bm{x}=\mathbf{T}\bm{s}=
\begin{bmatrix}\bm{t}_1&\ldots& \bm{t}_K\end{bmatrix}
\begin{bmatrix}s_1\\\vdots\\s_K\end{bmatrix} 
\label{eq:SM_2}
\end{equation}
where $\mathbf{T}\in \mathbb{C}^{K\times K}$ is the multi-user precoding matrix and $\bm{t}_i \in \mathbb{C}^{K\times 1}$ is the beamforming vector used to transmit $s_i$. Even though a per-TX power constraint is the most relevant power constraint in the multicell setting, we consider a sum power constraint $\norm{\mathbf{T}}_{\Fro}^2=P$. We also assume for simplicity and symmetry that all data streams are allocated with an equal amount of power so that $\bm{t}_i=\sqrt{P/K}\bm{u}_i$ with $\norm{\bm{u}_i}^{2}=1$. These choices can be done without restricting the scope of this work because they do not have any impact on the number of  DoFs\footnote{Indeed, it is always possible to scale the total power used when considering the sum power constraint so as to fulfill the per-TX power constraint without impacting the number of DoFs. Similarly, optimally allocating the power does not change the number of DoFs.}.  
We will study the ergodic rate averaged over the random codebooks~$\mathcal{W}_i^{(j)}$ used for the CSI Random Vector Quantization (RVQ), as detailed in Subsection~\ref{se:system_model:CSI}. The ergodic rate for RX~$i$ reads then as
\begin{equation}
R_i(P)\triangleq\E_{\mathbf{H},\{\mathcal{W}_i^{(j)}\}_{i,j}}\left[\log_2\left(1+\frac{|\bm{h}_i^{\He}\bm{t}_i|^2}{1+\sum_{\ell=1,\ell\neq i}^K|\bm{h}_{i}^{\He}\bm{t}_\ell|^2}\right)\right].
\label{eq:SM_3}
\end{equation}
To achieve the maximal number of DoFs we aim at removing completely the interference at all the RXs, i.e., at having
\begin{equation}
\forall i\in\{1,\ldots,K\},\quad \sum_{\ell=1,\ell\neq i}^K |\bm{h}_i^{\He}\bm{t}_{\ell}|^2=0.
\label{eq:SM_4}
\end{equation}
From \eqref{eq:SM_4} and the equal power allocation, there is no coupling between the optimizations of the beamforming vectors~$\bm{t}_i$ which can then be carried out in parallel. The number of DoFs achieved \emph{at RX~$i$} is defined as
\begin{equation}
\DoF_i\!\triangleq\!\lim_{P\rightarrow \infty}\!\frac{R_i(P)}{\log_2(P)}\!.
\label{eq:SM_5}
\end{equation}
and the total number of DoFs is $\DoF\triangleq \sum_{i=1}^K\DoF_i$. From the above definition of the number of DoFs and definition~\eqref{eq:SM_3}, we can directly obtain that $\forall i\in\{1,\ldots,K\}$,
\begin{equation} 
\DoF_i=1-\lim_{P\rightarrow \infty}\E_{\mathbf{H},\{\mathcal{W}_i^{(j)}\}_{i,j}}\left[\frac{\log_2\left(\sum_{\ell\neq i} |\bm{h}_i^{\He}\bm{t}_{\ell}|^2\right)}{\log_2(P)}\right].
\label{eq:SM_6}
\end{equation} 

\subsection{Distributed CSI}\label{se:system_model:CSI}

\subsubsection{CSI Scaling Coefficients}\label{se:system_model:CSI:Scaling}
We assume a limited CSI setting where channel estimate inaccuracies are modeled using quantized feedback. Furthermore, a \emph{distributed} CSI model is defined here in the sense that each TX has its own individual estimate of the normalized channel $\tilde{\bm{h}}_i$ to RX~$i$. Moreover, the estimates for the different channel vectors $\tilde{\bm{h}}_i$ are also a priori of different qualities at each TX, i.e., quantized with codebooks of different sizes. We denote by $\tilde{\bm{h}}^{(j)}_i$ the estimate of the normalized channel vector~$\tilde{\bm{h}}_i$ acquired at TX~$j$. The quantized feedback consists of $B^{(j)}_i$~bits which are used to index a vector in the codebook~$\mathcal{W}_i^{(j)}$ made of~$2^{B^{(j)}_i}$ elements. We also define $\tilde{\mathbf{H}}^{(j)}\triangleq[\tilde{\bm{h}}^{(j)}_1,\ldots,\tilde{\bm{h}}^{(j)}_K]^{\He}$ as the estimate of the total normalized multi-user channel at TX~$j$.

This setting arises in the context of multi-TX cooperation (e.g. Network MIMO \cite{Gesbert2010}) where either $(i)$ all TXs obtain a version of the whole CSI matrix through independent feedback channels (in which case the quality of the uplink feedback channel determines the quality of the individual CSI estimates) or $(ii)$ each TX obtains some portion of the CSI and exchange it through limited rate links or/and with some latency to the other TXs.

In the conventional MIMO BC, it is shown in~\cite{Jindal2006} that the number of quantization bits should scale indefinitely with the logarithm of the SNR in order to achieve a strictly positive number of DoFs when using ZF precoding. Thus, we also focus on the \emph{scaling in the logarithm of the SNR} of the number of quantization bits of all the channel estimates. We introduce the \emph{CSI scaling matrix} $\bm{\alpha}\in \mathbb{R}^{K\times K}$ with its $(i,j)$-th element defined as
\begin{equation}
\alpha_i^{(j)}\triangleq\lim_{P\rightarrow \infty}\frac{B^{(j)}_i}{(K-1)\log_2(P)}.
\label{eq:SM_scaling}
\end{equation}
Hence, $\alpha_i^{(j)}$ denotes the scaling of the number of bits used to describe the channel of user~$i$ at TX~$j$. Since~$B^{(j)}_i$ is a design parameter, the limit in \eqref{eq:SM_scaling} can be seen to always exist. We furthermore assume that the CSI scaling matrix~$\bm{\alpha}$ is known to all the TXs.

\emph{Remark: We will always consider for notational clarity $\alpha_i^{(j)}\in[0,1]$ as the range of interest. This follows from the fact that if~$\alpha_i^{(j)}=1$, it then holds~$|\bm{h}_i^{\He}\bm{t}_{\ell}^{(j)}|^2=O(1)$ for~$\ell\neq i$\cite{Jindal2006}. The accuracy of CSI resulting from a CSI scaling coefficient equal to one is sufficient for the interference to remain bounded. Thus, increasing the number of CSI feedback bits to get~$\alpha_i^{(j)}>1$ does not increase the number of DoFs. This corresponds to a well known result for the conventional MIMO BC in \cite{Jindal2006}. It follows that in all the subsequent results, the scaling coefficients~$\alpha_i^{(j)}$ should be replaced by~$\min(\alpha_i^{(j)},1)$ so as to be valid for arbitrary values for the CSI scaling coefficients. This is not done to keep the notations as clear as possible. } 

\subsubsection{Random Vector Quantization for the DCSI-MIMO Channel}\label{se:system_model:CSI:RVQ}
We consider RVQ where random codebooks are used to quantize the channels. This follows a result in \cite{Jindal2006} for the conventional MIMO BC, stating that in the case of two antennas at the TX, no codebook can achieve a better number of DoFs than the number of DoFs achieved with RVQ. RVQ is also shown to be optimal for the point-to-point MIMO link as the number of antennas tends to infinity both at the TX and the RX~\cite{Santipach2009}. Finally, RVQ is interesting because it gives an achievable lower bound.

In most of the works regarding the conventional MIMO BC, a codeword $\bm{w}$ is selected for quantizing the unit-norm vector $\tilde{\bm{h}}_i$ if it maximizes the amplitude of the inner product $|\tilde{\bm{h}}_i^{\He}\bm{w}|$. However, in the DCSI-MIMO channel, this quantization scheme is less adequate because the objective is invariant by multiplication of the codeword by a unit-norm complex number. This represents a problem since a different estimate is received at each TX, and this phase invariance creates an ambiguity between the estimates. This is very harmful for the transmission scheme and, in fact, if such a quantization scheme is used, it can be easily shown that the channel estimate obtained is essentially useless for joint precoding. 

Thus, another quantization scheme is preferred and the quantized channel $\tilde{\bm{h}}_i^{(j)}$ is instead obtained in the optimum $L_2$ norm sense:
\begin{equation}
\tilde{\bm{h}}_i^{(j)}=\argmin_{\bm{w}\in\mathcal{W}_i^{(j)}}\|{\bm{w}}-\tilde{\bm{h}}_i\|.
\label{eq:SM_7}
\end{equation}
Using directly \eqref{eq:SM_7} leads to lower performance as the phase of the channel also impacts the performance, and not only the direction in a Grassmannian space. To recover similar performance as the quantization scheme conventionally used, we multiply all the elements of the codebook as well as all the normalized channels by a complex number so as to let their first coefficient be real valued. A detailed analysis of this quantization scheme is provided in Appendix~\ref{se:Appendix_RVQ}.

\subsection{Distributed Precoding}
In the DCSI-MIMO channel, each TX has a different estimate of the multi-user channel~$\mathbf{H}$ and controls only one antenna. Thus, each TX uses its CSI to compute a certain precoding matrix from which it extracts the coefficient corresponding to its antenna. We denote the overall multi-user precoder computed at TX~$j$ as~$\mathbf{T}^{(j)}\triangleq\begin{bmatrix}\bm{t}_1^{(j)}&\ldots&\bm{t}_K^{(j)}\end{bmatrix}$ where $\bm{t}_i^{(j)}$ is the beamforming vector designed to transmit symbol~$s_i$. 

Note that although a given TX~$j$ may compute the whole precoding matrix $\mathbf{T}^{(j)}$, only the $j$-th row~$\bm{e}_j^{\trans}\mathbf{T}^{(j)}$ will be used in practice, since TX~$j$ transmits only~$x_j=\bm{e}_j^{\trans}\mathbf{T}^{(j)}\bm{s}$. Finally, the effective precoder is then given by
\begin{equation}
\mathbf{T}\triangleq\begin{bmatrix}\bm{t}_1&\ldots&\bm{t}_K\end{bmatrix}\triangleq\begin{bmatrix}\bm{e}_1^{\trans}\mathbf{T}^{(1)}\\\bm{e}_2^{\trans}\mathbf{T}^{(2)}\\\vdots\\\bm{e}_K^{\trans}\mathbf{T}^{(K)}\end{bmatrix}.
\label{eq:SM_8}
\end{equation}  
The main elements of the transmission in the distributed CSI MIMO channel are illustrated in Fig.~\ref{system_model_ISIT_journal}.
\begin{figure}
\centering
\includegraphics[width=1\columnwidth]{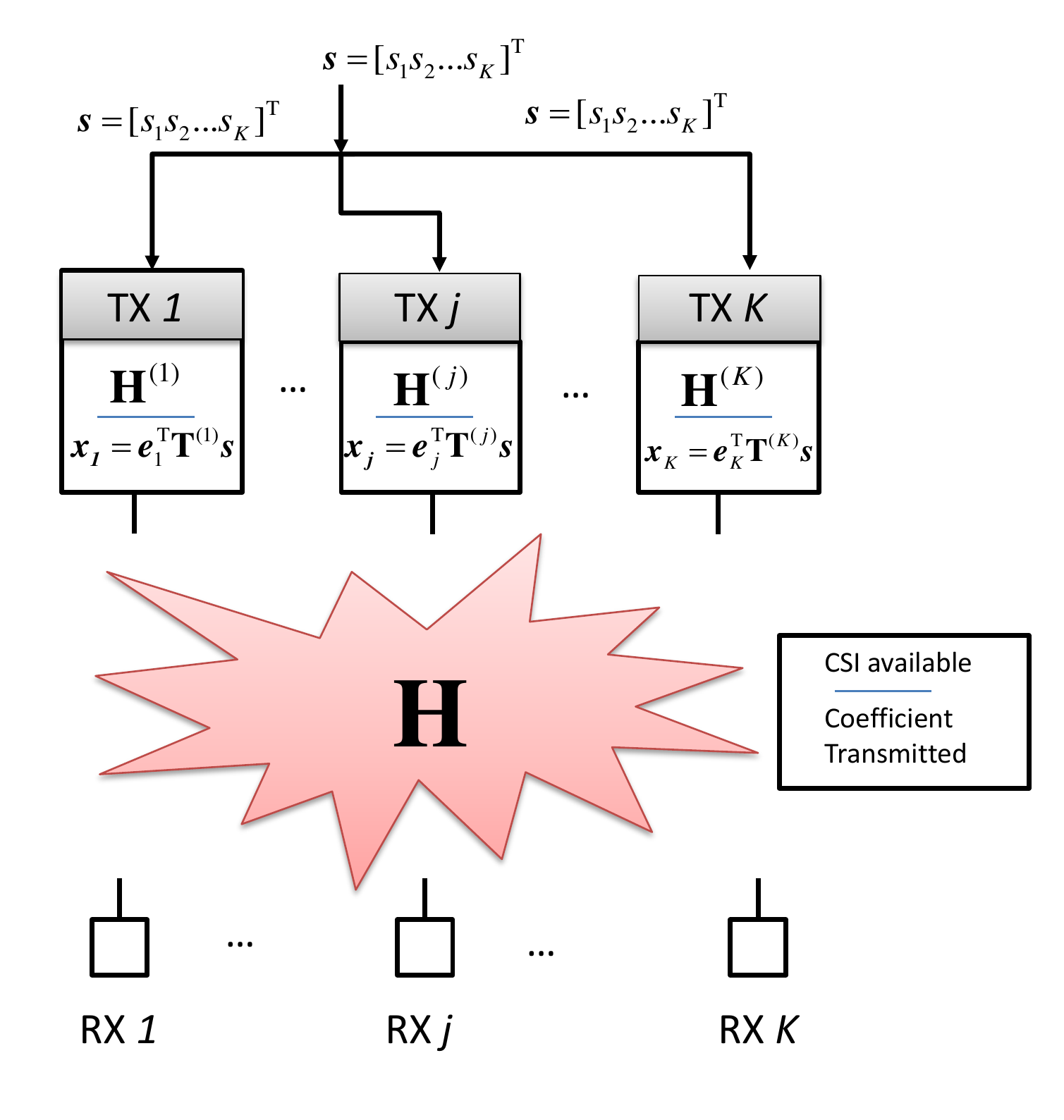}
\caption{Distributed precoding in the DCSI-MIMO channel.}
\label{system_model_ISIT_journal}
\end{figure}

\FloatBarrier
\section{Review of the Results in the Conventional MIMO BC}\label{se:MIMO_BC}
In this section, we recall briefly the main results from \cite{Jindal2006} on the number of DoFs achieved with finite rate feedback in the conventional MIMO BC. This will be helpful to understand the differences between the conventional MIMO BC and the distributed CSI setting which is the main focus of this work. 

Hence, we consider in this section a conventional MIMO BC where $M$~TXs are colocated and share the \emph{same} channel estimate. For this setting, we need to use different notations as previously introduced for the DCSI-MIMO channel. We denote by~$\hat{\bm{h}}_i$ the channel estimate of~$\tilde{\bm{h}}_i$ obtained with $B_i$~bits. Following \cite{Jindal2006}, the channel estimate is obtained from  
\begin{equation}
\hat{\bm{h}}_i=\argmax_{\bm{w}\in\mathcal{W}^{\BC}_i}|\bm{w}^{\He}\tilde{\bm{h}}_i|^2
\label{eq:BC_1}
\end{equation}
where~$\mathcal{W}^{\BC}_i$ is a random codebook containing~$2^{B_i}$ unit-norm vectors isotropically distributed in~$\mathbb{C}^{K\times 1}$. We provide now the main result.
\begin{theorem}\cite{Jindal2006}
In the MIMO BC with~$M$ antennas, if the channel estimate~$\hat{\bm{h}}_i$ is obtained from the quantization scheme~\eqref{eq:BC_1} with~$B_i=\alpha_i(M-1)\log_2(P)$, the number of DoFs achieved with ZF is given by
\begin{equation}
\DoF^{\BC}=\sum_{i=1}^M \alpha_i.
\label{eq:BC_2}
\end{equation}
\label{thm_DoF_BC}
\end{theorem}
This result was given in~\cite{Jindal2006} for $\alpha_i=\alpha$ but the extension to different~$\alpha_i$ follows directly from the proof in~\cite{Jindal2006}. The extension to Theorem~\ref{thm_DoF_BC} has been suggested in \cite{Vaze2010} where the same formula for the number of DoFs is derived in the case where DPC is used instead of ZF.


We will now derive the equivalent result of Theorem~\ref{thm_DoF_BC} for the DCSI-MIMO channel where the TXs do not share the same channel estimates.

\section{Zero Forcing in the DCSI-MIMO Channel with Two Users}\label{se:ZF_two}
As a starting point we consider the particular configuration with only two users. This setting is interesting for two main reasons. Firstly, the exposition is simpler in that case while most of the insights are the same as in the general case, and secondly this scenario makes it possible to obtain stronger results. 

In the conventional multiple-antenna BC with imperfect CSI, the number of DoFs with ZF has been derived and shown to be defined by the CSI scaling. In the DCSI-MIMO channel, the CSI scaling of each channel vector $\tilde{\bm{h}}_i$ is different at each TX. One central goal of our work consist in determining how the formula for the number of DoFs in the conventional MIMO BC generalizes to the DCSI-MIMO channel. This would then lead us to evaluate whether ZF is in that case a performing solution and if not, whether one can find better solutions.

\subsection{Conventional Zero Forcing}
In the DCSI-MIMO channel, the conventional ZF precoder is made of the beamformer $\bm{t}_i^{\mathrm{cZF}}\triangleq[\bm{e}_1^{\trans}\bm{t}_i^{\mathrm{cZF}(1)},\bm{e}_2^{\trans}\bm{t}_i^{\mathrm{cZF}(2)}]^ {\trans}$ to transmit $s_i$, with its elements defined in an intuitive way as
\begin{equation} 
\bm{t}_i^{\mathrm{cZF}(j)}\triangleq 
\sqrt{\frac{P}{2}}\frac{
\Pi_{\tilde{\bm{h}}_{\bar{i}}^{(j)}}^{\perp}\left(\tilde{\bm{h}}_{i}^{(j)}\right)}{\norm{\Pi_{\tilde{\bm{h}}_{\bar{i}}^{(j)}}^{\perp}\left(\tilde{\bm{h}}_{i}^{(j)}\right)}},\quad j\in\{1,2\}.
\label{eq:def_cZF_two}
\end{equation}
The interpretation behind conventional ZF is that each TX applies ZF using its own CSI implicitly assuming that the other TX shares the same CSI estimate. Our first result given in the following theorem relates the number of DoFs achieved with such a precoding strategy.

\begin{theorem}
Conventional ZF achieves the number of DoFs
\begin{equation}
\DoF^{c\mathrm{ZF}}=2\min_{i,j\in\{	1,2\}}\alpha_i^{(j)}.
\label{eq:thm_cZF_two}
\end{equation}
\label{thm_DoF_cZF_two}
\end{theorem}
\begin{IEEEproof}
A detailed proof is provided in Appendix~\ref{se:proof_thm_DoF_cZF}.
\end{IEEEproof} 
We can observe that in the case of distributed CSI, the number of DoFs is limited by the worst quality of the CSI across the channels to the RXs and across the TXs. Comparing this result with the number of DoFs achieved in a conventional MIMO BC given in Theorem~\ref{thm_DoF_BC}, it is remarkable that the number of DoFs at \emph{both} users is limited by the worst estimation error whether it is done relative to~$\tilde{\bm{h}}_1$ or~$\tilde{\bm{h}}_2$. This is contrast to the formula for the conventional MIMO BC in \eqref{eq:thm_cZF_two} where the accuracy of the estimation of~$\tilde{\bm{h}}_i$ impacts only the number of DoFs of RX~$i$. 

Note that when all the CSI scaling coefficients are equal, the setting considered is still different from the conventional multiple-antenna BC. Indeed, the estimates at the different TXs have statistically the same accuracy since the CSI scaling coefficients are equal, but the realizations of the estimation errors are still different. 

One can conclude from Theorem~\ref{thm_DoF_cZF_two} that the additional interference due to the CSI inconsistency between the TXs does not lead to any loss in number of DoFs compared to the conventional multiple-antenna BC if and only if the channel estimates are of the same quality.

\subsection{Robust Zero Forcing}
Robust precoding schemes have been derived in the literature either as statistical robust ZF precoder or precoder optimizing the worst case performance to reduce the harmful effect of the imperfect CSI. Since we consider the average sum rate, the most relevant approach is the statistical one. Thus, we model the quantization error at TX~$j$ by an additive white Gaussian noise $\bm{\Delta}^{(j)}\triangleq[\bm{\delta}_1^{(j)},\bm{\delta}_2^{(j)}]^{\He}$ of variance equal to $P ^{-\alpha_i^{(j)}}$ for the estimation error $\bm{\delta}_i^{(j)}$ resulting from the quantization of $\tilde{\bm{h}}_i$ at TX~$j$. The variance~$P ^{-\alpha_i^{(j)}}$ is obtained from the analysis of the scaling of the estimation error which is given in Appendix~\ref{se:Appendix_RVQ}.

The covariance matrix of the estimation error at TX~$j$ is then $\mathbf{R}_{\bm{\Delta}}^{(j)}\triangleq\E[\bm{\Delta}^{(j)}(\bm{\Delta}^{(j)})^{\He}]=\diag([P^{-\alpha_1^{(j)}},P^{-\alpha_2^{(j)}}])$. Using this model, we can extend the approach from \cite{Shenouda2006} and the beamformer transmitting symbol~$s_i$ at TX~$j$ is obtained from solving the following minimization:
\begin{equation}
\argmin_{\bm{t}_i}\E_{\bm{\Delta}^{(j)}}[\norm{\bm{e}_i-\tilde{\mathbf{H}}^{(j)}\bm{t}_i}^2],\quad\text{subject to $\|\bm{t}_i\|^2=\frac{P}{K}$}.
\end{equation}
Writing the Lagrangian of the minimization problem with the Lagrange variable $\lambda$ for the power constraint and taking the derivative according to $\bm{t}_i^*$ yields the equation
\begin{equation}
\left(\mathbf{R}^{(j)}_{\bm{\Delta}}+{\mathbf{H}}^{(j)\He}{\mathbf{H}}^{(j)}+\lambda\I\right)\bm{t}_i-{\mathbf{H}}^{(j)\He}\bm{e}_i=\bm{0}.
\end{equation}
The factor~$\lambda$ improves the performance at intermediate SNR by striking a compromise between the orthogonality constraint and the power consumption but it cannot improve the number of DoFs. Thus, we can let $\lambda$ be equal to zero and normalize the beamformer to fulfill the power constraint. The robust ZF beamformer transmitting symbol $s_i$ is denoted by $\bm{t}_{i}^{\mathrm{rZF}}\triangleq[\bm{e}_1^{\trans}\bm{t}_i^{\mathrm{rZF}(1)},\bm{e}_2^{\trans}\bm{t}_i^{\mathrm{rZF}(2)}]^ {\trans}$ and $\forall j\in\{1,2\}$ 
\begin{equation} 
\mathbf{t}_i^{\mathrm{rZF}(j)}\triangleq\sqrt{\frac{P}{K}}
\frac{(\mathbf{R}^{(j)}_{\bm{\Delta}}+{\mathbf{H}}^{(j)\He}{\mathbf{H}}^{(j)})^{-1}{\mathbf{H}}^{(j)\He}\bm{e}_i}{\left\|(\mathbf{R}^{(j)}_{\bm{\Delta}}+{\mathbf{H}}^{(j)\He}{\mathbf{H}}^{(j)})^{-1}{\mathbf{H}}^{(j)\He}\bm{e}_i\right\|}.
\label{eq:def_rZF}
\end{equation}
We then derive the number of DoFs achieved by this robust precoder. 
\begin{proposition}
The robust ZF precoder defined in \eqref{eq:def_rZF} achieves the same number of DoFs as conventional ZF.
\label{prop_rZF}
\end{proposition}
\begin{proof}
Considering strictly positive CSI scaling coefficients, the variances of the estimation errors tend to zero so that the inverse term in~\eqref{eq:def_rZF} can be approximated and we can write at RX~$\bar{i}$: 
\begin{align}
|\tilde{\bm{h}}_{\bar{i}}^{\He}\bm{t}_i^{(j)}|^2&=\frac{P}{K}\frac{|\tilde{\bm{h}}_{\bar{i}}^{\He}(\mathbf{R}^{(j)}_{\bm{\Delta}}+{\mathbf{H}}^{(j)\He}{\mathbf{H}}^{(j)})^{-1}{\mathbf{H}}^{(j)\He}\bm{e}_i|^2}{\left\|(\mathbf{R}^{(j)}_{\bm{\Delta}}+{\mathbf{H}}^{(j)\He}{\mathbf{H}}^{(j)})^{-1}{\mathbf{H}}^{(j)\He}\bm{e}_i\right\|^2}\\
&=\frac{P}{K}\frac{|\tilde{\mathbf{h}}_{\bar{i}}^{\He}
(\mathbf{H}^{(j)})^{-1}((\mathbf{H}^{(j)\He})^{-1}\mathbf{R}^{(j)}_{\bm{\Delta}}({\mathbf{H}}^{(j)})^{-1}+\I)^{-1}\bm{e}_i|^2}{\left\|(\mathbf{R}^{(j)}_{\bm{\Delta}}+{\mathbf{H}}^{(j)\He}{\mathbf{H}}^{(j)})^{-1}{\mathbf{H}}^{(j)\He}\bm{e}_i\right\|^2}\\
&=\frac{P}{K}\left(\frac{|\tilde{\bm{h}}_{\bar{i}}^{\He}
(\mathbf{H}^{(j)})^{-1}(\I-({\mathbf{H}}^{(j)\He})^{-1}\mathbf{R}^{(j)}_{\bm{\Delta}}({\mathbf{H}}^{(j)})^{-1})\bm{e}_i|^2}{\left\|(\mathbf{R}^{(j)}_{\bm{\Delta}}+{\mathbf{H}}^{(j)\He}{\mathbf{H}}^{(j)})^{-1}{\mathbf{H}}^{(j)\He}\bm{e}_i\right\|^2}+o(\norm{\mathbf{R}^{(j)}_{\bm{\Delta}}}^2_{\Fro})\right).
\end{align} 
The difference with conventional ZF is the term~$({\mathbf{H}}^{(j)\He})^{-1}\mathbf{R}^{(j)}_{\bm{\Delta}}({\mathbf{H}}^{(j)})^{-1}$ which can be shown to lead to no reduction of the interference and introduces actually an additional error term. Yet, it converges to zero as $P^{-\min(\alpha_1^{(j)},\alpha_2^{(j)})}$ since  $\mathbf{R}_{\bm{\Delta}}^{(j)}=\diag([P^{-\alpha_1^{(j)}},P^{-\alpha_2^{(j)}}])$. This is also the rate at which the remaining interference tends to zero when using conventional ZF. Thus, the regularizing term vanishes and the number of DoFs achieved is the same as conventional ZF. 
\end{proof}
Hence, even the existing designs of robust ZF precoders do not improve the number of DoFs in the DCSI-MIMO channel. Note that the extension of the definition of the statistical robust precoder as well as the extension of proposition~\ref{prop_rZF} to the general setting with $K$~users is trivial and will not be given explicitly.

\subsection{Beacon Zero Forcing}
Robust ZF schemes from the literature do not bring any DoFs improvement which leads to investigate other alternative schemes more adapted to the DCSI-MIMO channel. As a result, we now propose a modification of the conventional ZF scheme which improves the number of DoFs when the estimates for $\tilde{\bm{h}}_1$ and $\tilde{\bm{h}}_2$ are of different qualities. We call it \emph{Beacon ZF} (bZF) because it makes use of an arbitrary channel-independent vector known beforehand at both TXs (a \emph{beacon} signal). 

The beamformer used to transmit symbol~$s_i$ is then $\bm{t}_i^{\mathrm{bZF}}\triangleq[\bm{e}_1^{\trans}\bm{t}_i^{\mathrm{bZF}(1)},\bm{e}_2^{\trans}\bm{t}_i^{\mathrm{bZF}(2)}]^ {\trans}$, with its elements defined from
\begin{equation}  
\bm{t}_i^{\mathrm{bZF}(j)}\triangleq
\sqrt{\frac{P}{2}}\frac{
\Pi_{\tilde{\bm{h}}_{\bar{i}}^{(j)}}^{\perp}\left(\bm{c}_{i}\right)}{\norm{\Pi_{\tilde{\bm{h}}_{\bar{i}}^{(j)}}^{\perp}\left(\bm{c}_{i}\right)}}
\label{eq:def_bZF_two}
\end{equation}
where~$\bm{c}_{i}$ is any non-zero vector chosen beforehand and known at the TXs. Due to the isotropy of the channel, the choice of~$\bm{c}_{i}$ does not influence the performance of the precoder.

\begin{corollary}
The number of DoFs achieved with beacon ZF is 
\begin{equation}
\DoF^{\mathrm{bZF}}=\min_{j\in\{1,2\}}\alpha_1^{(j)}+\min_{j\in\{1,2\}}\alpha_2^{(j)}.
\label{eq:thm_bZF}
\end{equation}
\label{corollary_DoF_bZF_two_1}
\end{corollary}
\begin{proof}
The number of DoFs follows easily from Theorem~\ref{thm_DoF_cZF_two}. Indeed, when using beacon ZF, no error is induced by the projection of the direct channel which is replaced by a fixed given vector. In terms of number of DoFs, there is no difference between projecting the direct channel or any given vector. Thus, it is possible to apply the formula for the number of DoFs in Theorem~\ref{thm_DoF_cZF_two} considering that the direct channel is perfectly known, which yields the result.
\end{proof}
The key idea behind beacon ZF is to reduce the impact of the differences in CSI quality by using only the CSI necessary to fulfill the orthogonality constraint. Thus, the direct channel, which does not change the number of DoFs but only improves the finite SNR performance, is not used. It follows then that $\bm{t}_1^{\mathrm{bZF}}$ does no depend on the estimates of $\tilde{\bm{h}}_1$, and symmetrically $\bm{t}_2^{\mathrm{bZF}}$ does not depend on the estimates of $\tilde{\bm{h}}_2$. 

\subsection{Active-Passive Zero Forcing}\label{se:ZF_AP_two}
Beacon ZF improves the number of DoFs but it is still the worst CSI scaling across the TXs (although no longer across the RXs) which defines the number of DoFs. To improve further the number of DoFs, we propose a scheme called \emph{Active-Passive Zero Forcing (AP ZF)}. Assuming w.l.o.g. that $\alpha_{\bar{i}}^{(2)}\geq \alpha_{\bar{i}}^{(1)}$, AP ZF consists in the precoder whose beamformer $\bm{t}_i^{\mathrm{AP ZF}}$ transmitting symbol~$s_i$ is given by 
\begin{align}
\bm{t}_i^{\mathrm{APZF}}\!&\triangleq\!\sqrt{\frac{P}{2\log_2(P)}}\!\begin{bmatrix}1\\
-\frac{\{\tilde{\bm{h}}_{\bar{i}}^{(2)}\}_1}{\{\tilde{\bm{h}}_{\bar{i}}^{(2)}\}_2}\!
\end{bmatrix} \\
&=\!\sqrt{\!\frac{P(1\!+\!\rho_{i}^{(2)})}{2\log_2(P)}}\bm{u}_i^{\mathrm{APZF}}\!
\label{eq:def_apZF_two}
\end{align}
where 
\begin{equation}
\bm{u}_i^{\mathrm{APZF}}\triangleq 
\frac{\begin{bmatrix}1&
\frac{-\{\tilde{\bm{h}}_{\bar{i}}^{(2)}\}_1}{\{\tilde{\bm{h}}_{\bar{i}}^{(2)}\}_2}\!
\end{bmatrix} ^{\trans}}
{\left\|\begin{bmatrix}1&
\frac{-\{\tilde{\bm{h}}_{\bar{i}}^{(2)}\}_1}{\{\tilde{\bm{h}}_{\bar{i}}^{(2)}\}_2}\!
\end{bmatrix}^{\trans} \right\|}
\end{equation}
and $\rho_{i}^{(2)}\!\triangleq \!|\{\tilde{\bm{h}}_{\bar{i}}^{(2)}\}_1|^2/|\{\tilde{\bm{h}}_{\bar{i}}^{(2)}\}_2|^2$.

AP ZF is based on the idea that each beamforming vector has to fulfill only one orthogonality constraint so that only one available variable is necessary. Thus, one coefficient can be set to a constant while still fulfilling the ZF constraints. Moreover, the only way to achieve the number of DoFs stemming from the best CSI estimate is if TX~$2$ (which has the best knowledge of $\tilde{\bm{h}}_1$) can adapt to the coefficient transmitted at TX~$1$ to adjust its beamforming vector and improves the accuracy with which the interference are suppressed. This is possible only if TX~$2$ knows the transmit coefficient at TX~$1$. 

Using this precoding scheme, the number of DoFs is then given in the following proposition.

\begin{proposition} 
Active-Passive ZF achieves the number of DoFs:
\begin{equation}
\DoF^{\mathrm{APZF}}\geq\max_{j\in[1,2]}\alpha_1^{(j)}+\max_{j\in[1,2]}\alpha_2^{(j)}.
\end{equation}
\label{prop_DoF_apZF_two}
\end{proposition}

\begin{proof}
By symmetry, we consider w.l.o.g. the number of DoFs at RX~$1$, and we assume that the beamformers $\bm{t}_1$ and $\bm{t}_2$ are given by \eqref{eq:def_apZF_two}. We still assume w.l.o.g. that $\alpha_{1}^{(2)}\geq \alpha_{1}^{(1)}$, i.e., TX~$2$ has the best CSI over $\tilde{\bm{h}}_1$. From \eqref{eq:SM_6}, the number of DoFs at RX~$1$ is
\begin{equation}
\DoF_1\!
\!= 1-\lim_{P\rightarrow \infty}\frac{\E_{\mathbf{H},\{\mathcal{W}_{i,j}\}}\left[\log_2(|\bm{h}_1^{\He}\bm{t}_2|^2)\right]}{\log_2(P)}
\label{eq:lemma_apZF_two_1}
\end{equation}
We now focus on the interference term:
\begin{equation}
|\bm{h}_1^{\He}\bm{t}_2|^2=\frac{P}{2\log_2(P)}\left|\bm{h}_1^{\He}\begin{bmatrix}1\\ -\frac{\{\tilde{\bm{h}}_{1}^{(2)}\}_1}{\{\tilde{\bm{h}}_{1}^{(2)}\}_2}\end{bmatrix}\right|^2.
\label{eq:lemma_apZF_two_2}
\end{equation}
By construction, $\bm{t}_2$ is orthogonal to $\bm{h}_1^{(2)}$, so that 
\begin{align}
|\bm{h}_1^{\He}\bm{t}_2|^2 
&=\frac{P(1+\rho_{2}^{(2)})}{2\log_2(P)}\norm{\bm{h}_1}^2 \left|\Pi_{\tilde{\bm{h}}_1^{(2)}}^{\perp}(\tilde{\bm{h}}_1)^{\He}\bm{u}_2 +\left(\tilde{\bm{h}}_1^{(2)\He}\tilde{\bm{h}}_1\right)\tilde{\bm{h}}_1^{(2)\He}\bm{u}_2\right|^2\\ 
&=\frac{P(1+\rho_{2}^{(2)})}{2\log_2(P)}\norm{\bm{h}_1}^2 \sin^2(\tilde{\bm{h}}_1,\tilde{\bm{h}}_1^{(2)}).
\label{eq:lemma_apZF_two_3}
\end{align}
Inserting \eqref{eq:lemma_apZF_two_3} in the DoFs expression \eqref{eq:lemma_apZF_two_1} and using Proposition~\ref{App_log_distorsion} from Appendix~\ref{se:Appendix_RVQ} to bound the expectation of the sinus, we obtain 
\begin{align} 
\DoF_1
&\geq\lim_{P\rightarrow \infty} \frac{E_{\mathbf{H},\{\mathcal{W}_{i,j}\}}\left[-\log_2\left(\sin^2(\tilde{\bm{h}}_1,\tilde{\bm{h}}_1^{(2)})\right)\right]}{\log_2(P)}\label{eq:lemma_apZF_two_4_1}\\ 
&\geq\lim_{P\rightarrow \infty} \frac{B_1^{(2)}}{\log_2(P)}\label{eq:lemma_apZF_two_4_2}\\
&= \alpha_1^{(2)} \label{eq:lemma_apZF_two_4_3}
\end{align} 
which is the best scaling across the TXs.
\end{proof}
Comparing the number of DoFs achieved with AP ZF with the number of DoFs achieved when both TXs share the estimate of a channel vector with the highest accuracy gives the following result.
\begin{theorem} 
Active-Passive ZF achieves the same number of DoFs in the $2$-user DCSI MIMO channel as in the conventional MIMO BC where both TXs share the estimates with the highest CSI accuracy.
\label{thm_DoF_apZF_two}
\end{theorem}
\paragraph*{Improved scheme at finite SNR}
AP ZF allows to recover the number of DoFs which would have been achieved with the best CSI across the TXs. However, the choice of the coefficient used to transmit at TX~$1$ (with the lowest accuracy of the CSI) remains to be discussed. In fact, the beamformer can be multiplied arbitrarily by any unit-norm complex number without impacting the rate achieved so that only the power used at TX~$1$ needs to be decided. According to \eqref{eq:def_apZF_two}, the power used at TX~$1$ is set to $P/(2\log_2(P))$. 

The normalization by~$\log_2(P)$ is done because the fading coefficient $\{\tilde{\bm{h}}_{1}\}_2$ might have a very small amplitude. In this case it would be necessary for TX~$2$ to transmit with a very large power to fulfill the orthogonality constraint. To ensure that the interference are canceled for all channel realizations while respecting the power constraint, it is necessary to have the ratio between the power used at TX~$1$ and the sum power constraint tending to zero. The factor $\log_2(P)$ is used because it fulfills this property while not reducing the number of DoFs due to the partial power consumption. 

However, this comes at the cost of using only a small share of the available power, which is clearly inefficient and leads to a rate offset tending to minus infinity. To avoid this behavior, we propose that the TX with the worst CSI accuracy adapts its power consumption with respect to the channel realizations. In the following, we propose two possible solutions to improve the performance at finite SNR:
\begin{itemize}
\item 
Firstly, TX~$1$ can use its local CSI to normalize the beamformer which is then given by
\begin{equation}
\bm{t}_i^{\mathrm{APZF}}=\sqrt{\frac{P}{2}}\begin{bmatrix}\frac{1}{\sqrt{1+\rho_i^{(1)}}}\\
-\frac{\{\tilde{\bm{h}}_{\bar{i}}^{(2)}\}_1}{\sqrt{1+\rho_i^{(2)}}\{\tilde{\bm{h}}_{\bar{i}}^{(2)}\}_2}
\end{bmatrix}
\label{eq:def_hapZF}
\end{equation}
with $\rho_{i}^{(j)}\triangleq |\{\tilde{\bm{h}}_{\bar{i}}^{(j)}\}_1|^2/|\{\tilde{\bm{h}}_{\bar{i}}^{(j)}\}_2|^2$, for~$j=1,2$. This beamformer is not DoFs maximizing because the local CSI is used at TX~$1$ so that TX~$2$ does not any longer have an exact knowledge of the coefficient used to transmit at TX~$1$. Consequently, beamformer $\bm{t}_i^{\mathrm{APZF}}$ is not any longer orthogonal to $\tilde{\bm{h}}_{\bar{i}}^{(2)}$. Yet, this solution achieves good performance at intermediate SNR.

\item Another possibility is to assume that TX~$1$ receives the scalar $\rho_i^ {(2)}$ (or $\rho_i$) and use it to control its power. This means that TX~$2$ needs to share this scalar. This requires an additional feedback, but only a few bits are necessary to improve the performance at practical SNR.
\end{itemize}
\section{Zero Forcing in the DCSI-MIMO Channel for Arbitrary Number of Users}\label{se:ZF}

In this section, we will show how the main results can be generalized to arbitrary number of users. The same approach as in the case $K=2$ can be followed and we start by briefly generalizing to arbitrary number of users the precoding schemes previously described.

\subsection{Conventional Zero Forcing}
The conventional ZF precoder will be denoted as $\mathbf{T}^{\mathrm{cZF}}\triangleq[\bm{t}_1^{\mathrm{cZF}},\ldots,\bm{t}_K^{\mathrm{cZF}}]$  with $\bm{t}_i^{\mathrm{cZF}}\triangleq [\bm{e}_1^{\trans}\bm{t}_i^ {\mathrm{cZF}(1)},\bm{e}_2^{\trans}\bm{t}_i^ {\mathrm{cZF}(2)},\ldots,\bm{e}_K^{\trans}\bm{t}_i^ {\mathrm{cZF}(K)}]^{\trans}$ transmitting symbol $s_i$, and the beamformer $\bm{t}_i^{\mathrm{cZF}(j)}$ computed at TX~$j$ to transmit symbol~$i$ given by
\begin{equation} 
\bm{t}_i^ {\mathrm{cZF}(j)}\triangleq\sqrt{\frac{P}{K}} \frac{\Pi_{\bar{\mathbf{H}}_i^{(j)}}^{\perp}(\tilde{\bm{h}}_i^{(j)})}{\norm{\Pi_{\bar{\mathbf{H}}_i^{(j)}}^{\perp}(\tilde{\bm{h}}_i^{(j)})}}
\label{eq:def_cZF}
\end{equation}
with $\bar{\mathbf{H}}_i^{(j)}\triangleq [\tilde{\bm{h}}_1^{(j)},\ldots,\tilde{\bm{h}}_{i-1}^{(j)},\tilde{\bm{h}}_{i+1}^{(j)},\ldots,\tilde{\bm{h}}_K^{(j)}]$.

We can then generalize the results from~Theorem~\ref{thm_DoF_cZF_two} to an arbitrary number of users.
 
\begin{theorem}
In the DCSI-MIMO channel, the number of DoFs achieved with conventional ZF is equal to 
\begin{equation}
\DoF^{\mathrm{cZF}}=K\min_{i,j\in\{1,\ldots,K\}}\alpha_i^{(j)}.
\label{eq:thm_cZF}
\end{equation}
\label{thm_DoF_cZF}
\end{theorem}
\begin{IEEEproof}
A detailed proof is provided in Appendix~\ref{se:proof_thm_DoF_cZF}.
\end{IEEEproof} 
In Theorem~\ref{thm_DoF_cZF}, we have shown that the results concerning conventional ZF can be exactly generalized and the number of DoFs scales with the worst CSI accuracy across the TXs and the RXs. Indeed, the bad estimation of the channel to \emph{one} user at \emph{one} TX reduces the number of DoFs of \emph{all} the users. This is very pessimistic and represents a different behavior as in the conventional multiple-antennas BC. This can be observed by comparing the number of DoFs for the conventional MIMO BC in \eqref{eq:BC_2} with the formula for the number of DoFs in the DCSI-MIMO channel given in \eqref{eq:thm_cZF} when $\forall i,j=1,\ldots,K, \alpha_i^{(j)}=\alpha_i$, i.e., the CSI qualities are the same at all the TXs

\subsection{Beacon Zero Forcing}
The beacon ZF precoder is denoted as $\mathbf{T}^{\mathrm{bZF}}\triangleq [\bm{t}_1^{\mathrm{bZF}},\bm{t}_2^{\mathrm{bZF}}\ldots,\bm{t}_K^{\mathrm{bZF}}]$ with the beamformer $\bm{t}_i^{\mathrm{bZF}}\triangleq [\bm{e}^{\trans}_1\bm{t}_i^ {\mathrm{bZF}(1)},\bm{e}^{\trans}_2\bm{t}_i^ {\mathrm{bZF}(2)},\ldots,\bm{e}^{\trans}_K\bm{t}_i^ {\mathrm{bZF}(K)}]^{\trans}$ transmitting symbol~$s_i$. The beamformer~$\bm{t}_i^{\mathrm{bZF}(j)}$ computed at TX~$j$ to transmit symbol~$s_i$ is given by
\begin{equation}   
\bm{t}_i^ {\mathrm{bZF}(j)}\triangleq\sqrt{\frac{P}{K}} \frac{\Pi_{\bar{\mathbf{H}}_i^{(j)}}^{\perp}(\bm{c}_i)}{\norm{\Pi_{\bar{\mathbf{H}}_i^{(j)}}^{\perp}({\bm{c}}_i)}}
\label{eq:def_bZF}
\end{equation}
where $\bm{c}_{i}$ is any non-zero vector chosen beforehand and known at all TXs.

\begin{proposition}
The number of DoFs achieved with beacon ZF is equal to
\begin{equation}
\DoF^{\mathrm{bZF}}=\sum_{k=1}^K \;\;\min_{\underset{i\neq k}{i\in\{1,\ldots,K\},}}\;\;\min_{\underset{\ell\neq i}{\ell,j\in\{1,\ldots,K\},}} \alpha_{\ell}^{(j)}.
\label{eq:prop_DoF_bZF}
\end{equation} 
\label{prop_DoF_bZF}
\end{proposition}
\begin{proof}
To derive the number of DoFs at a RX~$k$, we need to compute the scaling of the interference at RX~$k$ stemming from the transmission to the $K-1$ other RXs. In the proof of Theorem~\ref{thm_DoF_cZF}, it is in fact the scaling of the interference resulting from the transmission of one stream which is calculated. To obtain the number of DoFs at one RX, the scaling of the interference resulting from the transmission of each of the $K-1$~interfering streams needs to be computed. This is represented by the first summation over~$i$. Determining the interference leaked by the transmission of symbol~$s_i$ using beacon ZF leads to the second minimum in the formula.
\end{proof} 
We have derived the number of DoFs for beacon ZF, but we will show in the following corollary that beacon ZF is only attractive in terms of number of DoFs in the two-user case.
\begin{corollary}
For $K\geq 3$, beacon ZF achieves the same number of DoFs as conventional ZF.
\label{corollary_DoF_bZF}
\end{corollary}
\begin{proof}
The result is easily obtained by studying the effect of the two successive minimums in \eqref{eq:prop_DoF_bZF}.
\end{proof} 

\subsection{Active-Passive Zero Forcing} 
The generalization of AP ZF is intuitive and consists simply, for the computation of each beamforming vector, in letting one TX arbitrarily fix its precoding coefficient while the other TXs adapt to this coefficient. Nevertheless, it requires the introduction of a few more notations. 

We define the ordered set $\mathbb{S}\triangleq\{n_1,\ldots,n_K\}$ as the set whose $i$-th element corresponds to the indice of the TX with fixed coefficient when transmitting the symbol $s_i$ (passive TX for $s_i$). We then introduce the (column) channel vector from TX~$\ell$ to all the RXs except the $i$-th RX:
\begin{equation}
\tilde{\bm{g}}_i^{(j)}(\ell)\triangleq[\{\tilde{\mathbf{H}}^{(j)}\}_{1,\ell}, \ldots ,\{\tilde{\mathbf{H}}^{(j)}\}_{i-1,\ell}, \{\tilde{\mathbf{H}}^{(j)}\}_{i+1,\ell}, \ldots ,\{\tilde{\mathbf{H}}^{(j)}\}_{K,\ell}]^{\trans}.
\end{equation}
Using the previous definition, we can then define
\begin{equation}  
\bar{\mathbf{H}}_i^{(j)}(n_i) \triangleq [\tilde{\bm{g}}_i^{(j)}(1), \ldots, \tilde{\bm{g}}_i^{(j)}(n_i-1),\tilde{\bm{g}}_i^{(j)}(n_i+1), \ldots,\tilde{\bm{g}}_i^{(j)}(K)]
\end{equation}
which represents the estimate at TX~$j$ of the multi-user channel from all the TXs except TX~$n_i$ to all the RXs except RX~$i$.
 
For a given set $\mathbb{S}$, we write $\mathbf{T}^{\mathrm{APZF}}(\mathbb{S})\triangleq[\bm{t}_1^{\mathrm{APZF}}(n_1),\bm{t}_2^{\mathrm{APZF}}(n_2),\ldots,\bm{t}_K^{\mathrm{APZF}}(n_K)]$ where the beamformer  $\bm{t}_i^{\mathrm{APZF}}(n_i)\triangleq[\bm{e}_1^{\trans}\bm{t}_i^{\mathrm{APZF}(1)}(n_i),\bm{e}_2^{\trans}\bm{t}_i^{\mathrm{APZF}(2)}(n_i),\ldots,\bm{e}_K^{\trans}\bm{t}_i^{\mathrm{APZF}(K)}(n_i)]^{\trans}$ transmits symbol $s_i$. The beamformer~$\bm{t}_i^{\mathrm{APZF}(j)}(n_i)$ computed at TX~$j$ to transmit symbol~$s_i$ is given by
\begin{equation}
\bm{t}_i^{\mathrm{APZF}(j)}(n_i)\!\triangleq\!\sqrt{\frac{P}{K\log_2(P)}}\bm{u}_i^ {\mathrm{APZF}(j)}(n_i) 
\label{eq:def_apZF}
\end{equation}
where we have defined
{
\begin{equation}
\bm{u}_i^ {\mathrm{APZF}(j)}\!(n_i)\!\triangleq\! [\;\check{u}_{1i}^{\mathrm{APZF}(j)}(n_i),\ldots,\check{u}_{n_i-1,i}^{\mathrm{APZF}(j)}(n_i) , 1 ,\check{u}_{n_i,i}^{\mathrm{APZF}(j)}(n_i) ,\ldots,\check{u}_{K-1,i}^{\mathrm{APZF}(j)}(n_i) ]^{\trans}
\end{equation}}
 with $\check{\bm{u}}_i^ {\mathrm{APZF}(j)}(n_i) \triangleq\begin{bmatrix}\check{u}_{1i}^{\mathrm{APZF}(j)}(n_i),\ldots, \check{u}_{K-1,i}^{\mathrm{APZF}(j)}(n_i)\end{bmatrix}^{\trans}\in \mathbb{C}^{K-1\time 1}$ and
\begin{equation}
\check{\bm{u}}_i^ {\mathrm{APZF}(j)}(n_i)\!\triangleq\! \frac{-\left(\bar{\mathbf{H}}_i^{(j)}(n_i)\right)^{-1}\tilde{\bm{g}}_i^{(j)}(n_i)}{\sqrt{1+\norm{\left(\bar{\mathbf{H}}_i^{(j)}(n_i)\right)^{-1}\tilde{\bm{g}}_i^{(j)}(n_i)}^2}}.
\end{equation}
Even though the notations are quite heavy, the intuition behind the construction of the precoder is exactly the same as for the two-user case. TX~$n_i$ is the \emph{passive} TX and transmits with a fixed coefficient $\sqrt{P/K\log_2(P)}$ while the other \emph{active} TXs then choose their coefficients in order to ZF the interference. This is obtained by setting their coefficients so as to fulfill~\eqref{eq:def_apZF}. The notational complexity comes only from the fact that we need to introduce a ``reduced'' channel without the direct channel as well as without the channel from the \emph{passive} TX.
\begin{proposition}
Active-Passive ZF with the set $\mathbb{S}=\{n_1,\ldots,n_K\}$ achieves the number of DoFs
\begin{equation}
\DoF^{\mathrm{APZF}}(\mathbb{S})=\sum_{k=1}^K \min_{\underset{i\neq k}{i\in\{1,\ldots,K\},}} \;\;\min_{\underset{\ell\neq i,j\neq n_{i}}{\ell,j\in\{1,\ldots,K\},}} \alpha_{\ell}^{(j)}.
\label{eq:prop_DoF_apZF}
\end{equation}  
\label{prop_DoF_apZF}
\end{proposition}
\begin{proof}
Due to the symmetry between the RXs, we will show the result only for the number of DoFs at RX~$k$. Let assume that AP ZF is used with the set $\mathbb{S}$.
To obtain the number of DoFs, we need to derive the scaling of the interference at RX~$i$ when all streams are transmitted using AP ZF. The first minimum of the DoFs formula follows from the summation over all the $K-1$~interfering streams. It remains then to determine the scaling of the interference resulting from the transmission of one given data symbol. 

TX~$j$ computes the beamformer $\bm{t}_{\ell}^ {\mathrm{APZF}(j)}(n_{\ell})$ according to~\eqref{eq:def_apZF}. This formula is similar to the one for conventional ZF so that the scaling of the remaining interference power can be derived with a proof very akin to that of Theorem~\ref{thm_DoF_cZF} which is omitted to avoid repetitions. Thus, the interference received at RX~$k$ due to the transmission of symbol~$s_i$ corresponds to the second minimum of the DoFs formula. This expression follows from the fact that the CSI at TX~$n_{\ell}$ and the CSI on the direct channel $\tilde{\bm{h}}_{\ell}$ are not used to design the beamformer transmitting~$s_{\ell}$.
\end{proof}
The number of DoFs given in Proposition~\ref{prop_DoF_apZF} is given by two successive minimizations. This is similar to beacon ZF at the difference that the index of one TX is not taken into account in the second minimization. This leads then to a larger number of DoFs. The formula for the number of DoFs depends on the set $\mathbb{S}$ but we will show that the optimal set is easily derived when the number of users is larger than~$4$.
\begin{corollary}
For $K\geq 4$ users, it is optimal in terms of number of DoFs to choose all the indices in $\mathbb{S}$ to be equal. Therefore, it is optimal to choose $n_i$ as the indice of the minimum over all the CSI scaling coefficients, and the number of DoFs reads as
\begin{equation}
\DoF^{\mathrm{APZF}}=K \min_{\underset{j\neq \argmin_{k}\min_{\ell}\alpha_{\ell}^{(k)}}{i,j\in\{1,\ldots,K\},}} \alpha_{i}^{(j)}.
\label{eq:corollary_DoF_apZF}
\end{equation}  
\label{corollary_DoF_apZF}
\end{corollary}
\begin{proof}
Similar to the proof of the corollary for Beacon ZF, the proof follows by studying the effect of the two successive minimums and for $K\geq 4$, it has for consequence that it is optimal to choose $\forall i,j,n_i=n_j$. 
\end{proof}
Exactly as in the two-user case, AP ZF leads to an improvement in number of DoFs but this comes at the cost of an unbounded negative rate offset. To improve on this feature, the percentage of the available power which is consumed by the TXs needs to be increased. The sames solutions as described for the two-user case in Subsection~\ref{se:ZF_AP_two} can be applied, i.e., either a heuristic power control or the transmission of a scalar to control the power. Note that the scalar can be transmitted by any of the other $K-1$ TXs and that one scalar needs to be transmitted for each stream. We refer to Subsection~\ref{se:ZF_AP_two} for more details.

\subsection{Discussion of the Results}
Altogether, we have shown in this section that the results for the two-user case given in Section~\ref{se:ZF_two} could generalize to an arbitrary number of users. However, the results suggest in all cases a fundamental lack of robustness of the performance as we increase the number of users. Indeed, with conventional ZF, a single inaccurate channel estimate can reduce the number of DoFs of all the users while the novel precoding schemes proposed can only cope with a few channel estimates being of insufficient quality. This shows the need for other methods to make the transmission more robust to imperfect distributed CSI when more than two-user are present.

\section{Precoding Using Hierarchical Quantization}\label{se:hierarchical_quantization}
In view of the rather pessimistic results in the previous section, we propose now an alternative method to make the transmission more robust to the CSI discrepancies. It consists in modifying the CSI quantization and using a Hierarchical Quantization (HQ) scheme to encode the CSI\cite{Boccardi2007,Zakhour2010a}.  

\subsection{Hierarchical Quantization}
Hierarchical quantization (or multi-resolution quantization) is a quantization scheme in which the information is encoded so that the original message can be decoded up to a number of bits depending on the quality of the feedback channel. The better the channel is, the more bits can be decoded. Thus, if one entity receives a codeword with a higher accuracy than another entity, and has the knowledge of the feedback qualities, it also knows what has been decoded at the other entity. Conversely, if one entity can detect the feedback information at a given resolution level but knows that another entity can decode the same information at a higher resolution level, it can use its individual decoded codeword to form a limited set of guesses around it as to which higher resolution codeword may have been detected at the other TX.

In our setting, it means that each TX can decode the CSI feedback up to a certain number of bits depending on the quality of the feedback link. If TX~$j_1$ receives a CSI of better quality than another TX~$j_2$, it can decode more bits from the CSI and can get the knowledge of the CSI at TX~$j_2$ with less decoded bits. Note that this implies that two TXs with the same CSI quality have the \emph{same} codebook and thus exactly the same realization for the channel estimation error. This is in contrast to what has been considered in the previous sections.

We wish to continue using the properties of RVQ so that we need to design \emph{hierarchical random codebooks}, i.e., codebooks fulfilling the properties of both kinds of codebooks. Since this is not the main focus of the work, we just briefly describe a possible method to construct such codebooks and the quantization scheme associated. 

We start by considering a random codebook of size corresponding to the best accuracy, say $2^{\ell_{\max}}$. This random codebook is then divided into two random codebooks containing each half the elements. This process is then applied on the two smaller codebooks obtained until having $2^{\ell_{\max}}$ codebooks of one element. In each of the sub-codebooks of different sizes created, we pick randomly one elements to be the \emph{representative} of this codebook. 

Once the quantized vector maximizing the figure of merit has been chosen among the~$2^{\ell_{\max}}$ vectors, the encoding can be easily done. The chosen vector belongs to one set of each size and the encoding bits are used to select among the two possible choices, the set to which the quantized vector belongs. 

The decoding step works as follows. The first bit denotes one of the two codebooks of size $2^{\ell_{\max}-1}$, the second bit denotes one of the two codebooks of size $2^{\ell_{\max}-2}$ inside this codebook, and so on, until the last bit is decoded. Once this is done, the codeword decoded is chosen to be the \emph{representative} codeword of the obtained codebook.

It is then easily verified that the proposed quantization scheme has the hierarchical properties desired.

\subsection{Conventional Zero Forcing with Hierarchical Quantization}

In the previous sections, we have shown that the quality of the estimation of one channel $\tilde{\bm{h}}_i$ to one given RX had an impact on the number of DoFs achieved \emph{at all RXs}. This is a surprising property which follows from the particular structure of the DCSI-MIMO channel where the consistency between the transmissions of the different TXs is critical. We will show how the hierarchical quantization described above can be used to avoid this very inefficient property. 

In the following, we will consider a particularly simple use of hierarchical quantization consisting in letting all the TXs designing the beamforming vector use only the part of the CSI which is common to all the TXs, and simply "forget" about the more accurate CSI knowledge. We then obtain a CSI configuration where all the TXs share the same CSI and the number of DoFs can be obtained from Theorem~\ref{thm_DoF_BC}. 

\begin{theorem}
The number of DoFs achieved using Conventional ZF with hierarchical quantization is
\begin{equation}
\DoF^{\mathrm{cZF}}=\sum_{i=1}^K \min_{j\in\{1,\ldots,K\}} \alpha_{i}^{(j)}.
\label{eq:thm_DoF_cZF_HQ}
\end{equation} 
\label{thm_DoF_cZF_HQ} 
\end{theorem}

Using HQ as described, i.e., using only the estimate of a channel vector~$\tilde{\bm{h}}_i$ common to all the TXs, follows from the observation that the worst estimation error of~$\tilde{\bm{h}}_i$ limits in any case the number of DoFs at RX~$i$. Thus, using only the common part of the estimate of~$\tilde{\bm{h}}_i$ does not reduce the number of DoFs at RX~$\tilde{\bm{h}}_i$. Yet, it leads to an improved consistency between the beamformers computed at the TXs. This has for consequence that the error in the estimate of the channel $\tilde{\bm{h}}_i$ only impacts the number of DoFs at RX~$i$ and not at the other RXs.

Note that the proposed scheme using HQ is very simple and more gains could certainly be obtained with a more sophisticated use of the additional CSI knowledge available at some TXs.


\subsection{Active-Passive Zero Forcing with Hierarchical Quantization}
Hierarchical quantization is used for AP ZF in the same way as for Conventional ZF. This consists in using the CSI which is common to all the active TXs considered in the definition of the beamformer in~\eqref{eq:def_apZF}.
\begin{proposition}
The number of DoFs achieved using Active-Passive ZF with Hierarchical Quantization and the set $\mathbb{S}$ is
\begin{equation}
\DoF^{\mathrm{APZF}}(\mathbb{S})=\sum_{k=1}^K \min_{\underset{i\neq k}{i\in\{1,\ldots,K\},}}\;\;\min_{\underset{j\neq n_i}{j\in\{1,\ldots,K\},}} \alpha_{k}^{(j)}.
\label{eq:prop_DoF_apZF_HQ}
\end{equation}  
\label{prop_DoF_apZF_HQ}
\end{proposition}
The two successive minimums come from the fact that it is not the same TX which is \emph{passive} for the different streams. It is clear from \eqref{eq:prop_DoF_apZF_HQ} that it is optimal to choose all the~$n_i$ to be equal for~$K\geq 3$. However, the indice of the optimal passive TX, which we denote by $n_{\mathrm{HQ}}$, is now different from the case without HQ. It is easily obtained by looking for the passive TX bringing the largest improvement in number of DoFs:
\begin{equation}
n_{\mathrm{HQ}}\triangleq \argmax_{n\in\{1,\ldots,K\}} \sum_{k=1}^K\min_{\underset{j\neq n}{j\in\{1,\ldots,K\},}}\alpha_k^{(j)}.
\label{eq:n_HC}
\end{equation}
The maximum number of DoFs using AP-ZF with HQ follows then directly.
\begin{proposition}
For~$K\geq 3$, it is optimal to choose the passive TX to be TX~$j$ with~$j=n_{\mathrm{HQ}}$ defined in \eqref{eq:n_HC}, for all the data streams. The number of DoFs achieved with Active-Passive ZF based on Hierarchical Quantization is then equal to 
\begin{equation}
\DoF^{\mathrm{APZF}}=\sum_{i=1}^K \min_{\underset{j\neq n_{\mathrm{HC}}}{j\in\{1,\ldots,K\},}} \alpha_{i}^{(j)}.
\label{eq:thm_DoF_apZF_HQ}
\end{equation}
\label{prop_DoF_apZF_HQ}
\end{proposition}

\section{DoF Optimal Sharing of the Feedback under a Total Feedback Constraint}\label{se:sharing}
In this section, we consider the opposite side of the problem which consists in deriving how to distribute a maximum number~$B$ of feedback bits across the TXs and the channel vectors so as to maximize the number of DoFs. Since our focus remains on the number of DoFs and considering previous results, it is meaningful to introduce $\gamma\triangleq \lim_{P\rightarrow \infty} B/\log_2(P)$ which we call the \emph{total feedback scaling}. 

Thus, we consider a constraint on the sum of the scaling coefficients of the total feedback transmitted through the multi-user channel feedback: 
\begin{equation}
\sum_{i,j\in\{1,\ldots,K\}}\alpha_i^{(j)}\leq \gamma.
\label{eq:feedback_constraint}
\end{equation}
We study first conventional ZF before extending the results to Active-Passive ZF. To optimize the CSI allocation efficiently, it becomes necessary to also optimize the number of users being served, which means that time sharing will this time be explicitly considered.

\subsection{Conventional Zero Forcing}
\begin{proposition}
With conventional ZF (with or without Hierarchical Quantization), it is optimal in terms of number of DoFs to share equally the number of bits across the TXs and across the channels to quantize and to let the number of TX being actually transmitting be equal to $n$ for $\gamma \in [n(n-1)^2,(n+1)n^2]$. It follows that the optimal number of DoFs using Conventional ZF is equal to 
\begin{equation}
\begin{cases}
\DoF^{\mathrm{cZF}}=\gamma/(n(n-1)),&\text{  if $\gamma \in [n(n-1)^2,n^2(n-1)]$}\\
\DoF^{\mathrm{cZF}}=n,&\text{  if $\gamma \in [n^2(n-1),(n+1)n^2]$}.
\end{cases}
\end{equation}
\label{prop_DoF_sharing_cZF}
\end{proposition}
\begin{proof}
We study first the case without HQ. Since the number of DoFs scales as the worst CSI scaling across the TXs and the channel vectors, it is clearly optimal to have the same CSI accuracy at all the TXs and for all the channel vectors. To achieve a number of DoFs of $\alpha$ at $n$~RXs, the number of bits to quantize a channel vector has to be equal to $\alpha(n-1)\log_2(P)$, where $n$ is the number of transmitting TXs. Hence, the total feedback in the channel is given by $n^2\alpha(n-1)\log_2(P)$ when considering the $n$~estimates needed at the $n$~TXs.

Let's assume that $n$~TXs are serving $n$~RXs with the maximal feedback scaling~$\gamma$, we obtain that $\alpha=\gamma/(n^2(n-1))$. For $\gamma\leq n^2(n-1)$ the number of DoFs achieved at the RXs is lower or equal to one so that the sum number of DoFs is equal to $n\alpha=\gamma/(n(n-1))$. For $\gamma\geq n^2(n-1)$, the number of DoFs at each RX reaches its maximal value of one and the sum number of DoFs is equal to $n$. Comparing the sum number of DoFs achieved by two successive configurations, with respectively $n$ and $n+1$ users served, leads to the value of $\gamma$ given in the proposition as switching point between the configurations.

When HQ is used, the number of DoFs still scales as the minimum over the CSI scaling across the TXs so that it is still optimal to let all the TXs have the same CSI scaling.
\end{proof}
Using HQ does not increase the number of DoFs when the CSI configuration can be optimized. However, many more configurations are optimal as the CSI can be allocated indifferently to any channel vector as long as the scaling of the CSI does not exceed one and all the TXs receive the same CSI. 

\begin{figure}
\centering
\includegraphics[width=1\columnwidth]{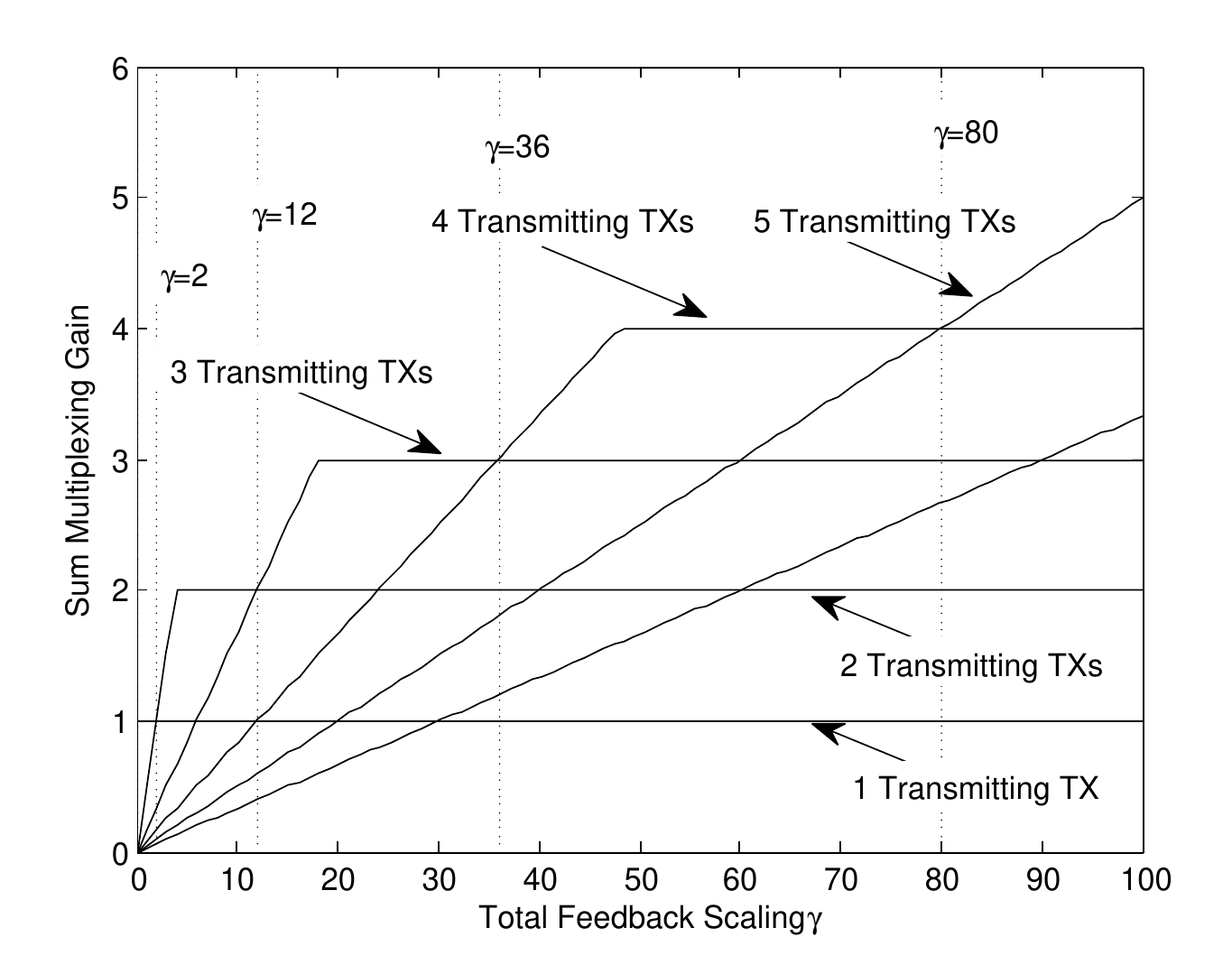}
\caption{Degrees of Freedom as a function of the total feedback scaling $\gamma$ for different number of users.}
\label{fig_DoF_sharing_cZF}
\end{figure}

The results from Proposition~\ref{prop_DoF_sharing_cZF} are very intuitive, yet the formula is not very enlightening and the intuition is better understood in a plot of the number of DoFs with optimal CSI sharing. Thus, we plot in Figure~\ref{fig_DoF_sharing_cZF} the number of DoFs in terms of the total feedback scaling $\gamma$ for different numbers of transmitting TXs. The parts with a positive slope correspond to values of $\alpha$ smaller than one while the flat parts correspond to a saturation of the number of DoFs, i.e., $\alpha\geq 1$. 

The values of $\gamma$ corresponding to the saturation of the number of DoFs and to the activation of an additional user, respectively, are given in Appendix~\ref{se:App_DoF_sharing}. When $n$~TXs are transmitting, the slope of the number of DoFs as a function of $\gamma$ is known to be equal to $1/(n^2(n-1))$ and we can observe in the figure how the values for $\gamma$ given in the proposition fit with the observation in terms of saturation and intersection of the curves.

It is possible to observe that the saturated parts are optimal for some values of $\gamma$. This follows from the fact that using an additional TX induces an increase of the feedback necessary (lower slope in the figure). Thus, a possibly large increase in $\gamma$ is necessary before reaching the point where it starts being more interesting to serve the additional RX and use an additional TX.

\subsection{Extension to Active-Passive Zero Forcing}
Our analysis for conventional ZF can be extended to Active-Passive ZF without difficulty. The only difference consists in the number of bits necessary to achieve a scaling of~$\alpha$ which is then $n(n-1)^2\alpha\log_2(P)$ instead of $n^2(n-1)\alpha\log_2(P)$ since one TX (passive TX) does not need to be shared any CSI. Thus, it holds that $\alpha=\gamma/(n(n-1)^2)$ which leads to the following result.

\begin{proposition}
When using Active-Passive ZF (with or without HQ), it is optimal to share equally the number of bits across the active TXs and across the channels, and to let the number of transmitting TXs be equal to $n$ for $\gamma \in [n(n-1)^2,(n+1)n^2]$. It follows that the optimal number of DoFs is equal to 
\begin{equation}
\begin{cases}
\DoF^{\mathrm{APZF}}=\gamma/(n-1)^2,&\text{  if $\gamma \in [(n-1)^3,n(n-1)^2]$}\\
\DoF^{\mathrm{APZF}}=n,&\text{  if $\gamma \in [n(n-1)^2,n^3]$}.
\end{cases}
\end{equation}
\label{prop_DoF_sharing_apZF}
\end{proposition}
The proof and the plot of the number of DoFs in terms of the total feedback scaling $\gamma$ follow both the same pattern as conventional ZF and are omitted to avoid repetition.

The general insight behind those results is that it is better to achieve the maximal number of DoFs at less users instead of serving more users with a lower number of DoFs. This is an intuitive consequence of the very quick increase of the size of the aggregate feedback required in terms of the number of TXs used. 
%

\section{Simulations}

\subsection{In the Two-User Case}\label{se:simulations_two}
We consider two models for the imperfect channel CSI, a statistical model and RVQ. 

In the statistical model, the quantization error is modeled by adding a Gaussian i.i.d. quantization noise to the channel with the covariance matrix at TX~$j$ equal to $\diag([P^{-\alpha_1^{(j)}},P^{-\alpha_2^{(j)}}])$. This corresponds to the scaling in~$P$ of the variance provided in Proposition~\ref{App_distorsion} of Appendix~\ref{se:Appendix_RVQ}. The Gaussian distribution maximizes the entropy for the given variance~\cite{Cover2006} so that we will obtain a priori a lower bound for the performance. Yet, it is expected that only the scaling of the variance will have an impact so that the statistical model should be accurate. The averaging is then done over $10000$~realizations.

In the RVQ, we consider a given number of feedback bits and we average over $100$ random codebooks and $1000$ channel realizations. In the simulations, we consider the following precoders: ZF with perfect CSI, conventional ZF [cf.~\eqref{eq:def_cZF_two}], Beacon ZF [cf.~\eqref{eq:def_bZF_two}], and Active-Passive ZF [cf.~\eqref{eq:def_apZF_two}] with heuristic power control and with $3$-bits power control.

In Fig.~\ref{Rate_1_05_0_07_stat}, we consider the statistical model with the CSI scaling $[\alpha_1^ {(1)},\alpha_1^ {(2)}]=[1,0.5]$ and $[\alpha_2^{(1)},\alpha_2^{(2)}]=[0,0.7]$. To emphasize the number of DoFs (i.e., the slope of the curve in the figure), we let the SNR grow large. As expected theoretically, conventional ZF scales with the worst accuracy and saturates at high SNR, while Beacon ZF has a positive slope and Active-Passive ZF performs closer to perfect ZF with a slope only slightly smaller than the optimal one.

\begin{figure}
\centering
\includegraphics[width=1\columnwidth]{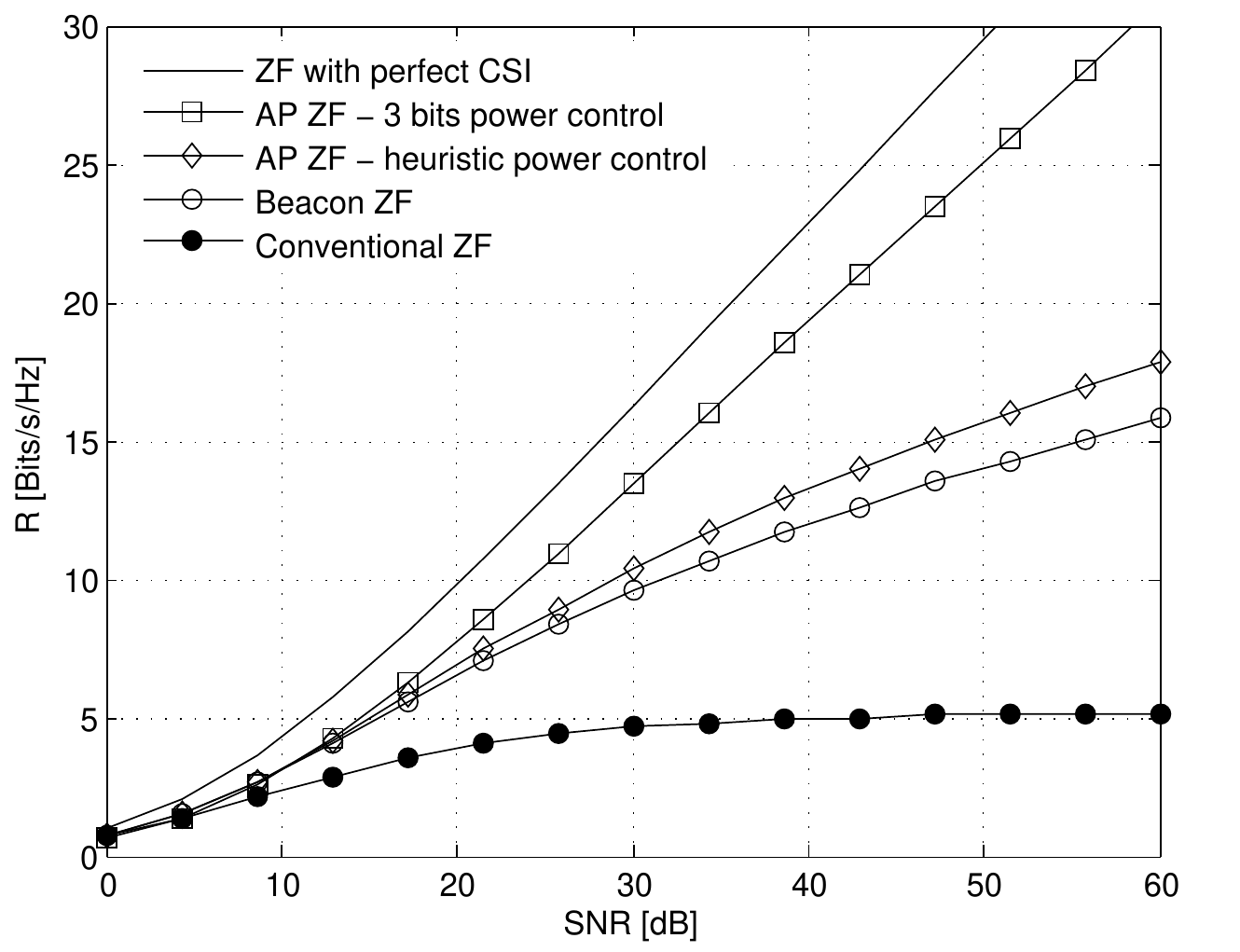}
\caption{Sum rate in terms of the SNR with a statistical modeling of the error from RVQ using $[\alpha_1^ {(1)},\alpha_1^ {(2)}]=[1,0.5]$ and $[\alpha_2^{(1)},\alpha_2^{(2)}]=[0,0.7]$.}
\label{Rate_1_05_0_07_stat}
\end{figure}

In Fig.~\ref{Rate_bits_6_3_3_6_RVQ}, we plot the sum rate achieved with the CSI feedback $[B_1^{(1)},B_1^{(2)}]=[6,3]$ and $[B_2^{(1)},B_2^{(2)}]=[3,6]$ using RVQ. From the theoretical analysis, the number of DoFs should be equal to zero for all the precoding schemes since the number of feedback bits used does not increase with the SNR. This is confirmed by the saturation of the sum rate as the SNR increases. Yet, the saturation occurs at a higher SNR for Beacon ZF compared to conventional ZF, and at an even higher SNR for Active-Passive ZF. This translates into an improvement of the sum rate at intermediate SNR. 


\begin{figure}
\centering
\includegraphics[width=1\columnwidth]{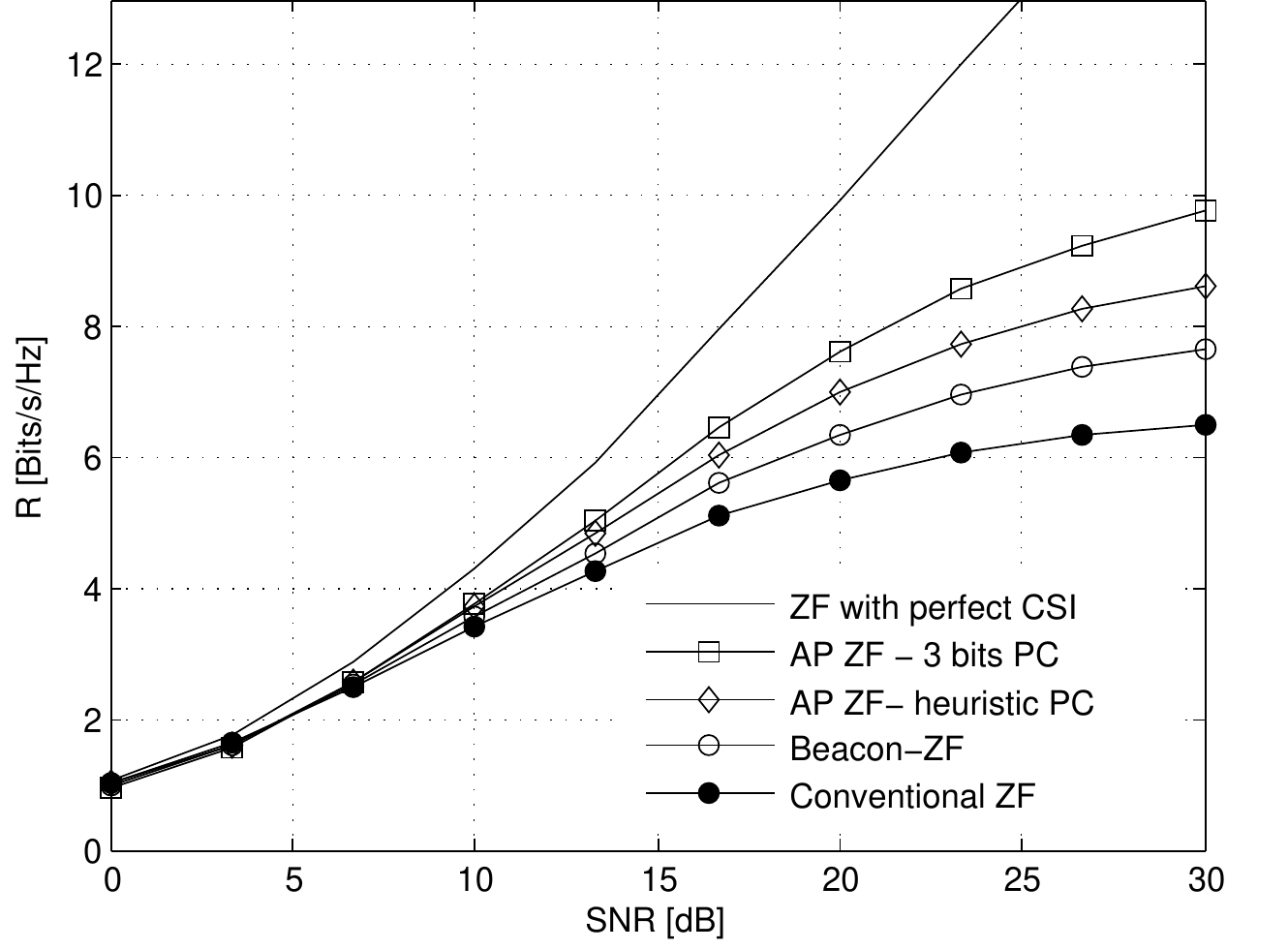}
\caption{Sum rate in terms of the SNR with RVQ using $[B_1^{(1)},B_1^{(2)}]=[6,3]$ and $[B_2^{(1)},B_2^{(2)}]=[3,6]$.}
\label{Rate_bits_6_3_3_6_RVQ}
\end{figure}

\FloatBarrier
\subsection{With Arbitrary Number of Users} \label{se:simulations_general}

For the simulations with arbitrary number of users, only the statistical model described in the previous paragraph for the two-user case is considered. To model easily the use of Hierarchical Quantization, we simply consider that a TX has the knowledge of the channel estimate at another TX if this TX receives a feedback concerning this channel vector with a lower CSI scaling coefficient. Since we have derived that Beacon ZF [Cf.~\eqref{eq:def_bZF}] does not bring any improvement in number of DoFs for $K\geq3$, we will consider in the figures only conventional ZF [Cf.~\eqref{eq:def_cZF}] and Active-Passive ZF [Cf.~\eqref{eq:def_apZF}] where the transmission of $3$-bits to the \emph{passive} TX is allowed for every beamforming vector. For both precoding schemes, we will furthermore consider both the case of Hierarchical Quantization with random codebooks and conventional RVQ.

We consider the performance achieved with an arbitrary chosen CSI scaling matrix to verify that the precoding schemes behave as expected. Thus, we consider $K=7$ users and we set all the elements of the CSI scaling matrix $\bm{\alpha}$ equal to $1$ at the exception of two coefficients corresponding to different TXs and RXs set to $0$ and $0.3$, respectively. The CSI scaling matrix is given explicitly in Appendix~\ref{se:App_CSI_matrix} as well as the number of DoFs obtained analytically for that setting. 


\begin{figure}
\centering 
\includegraphics[width=1\columnwidth]{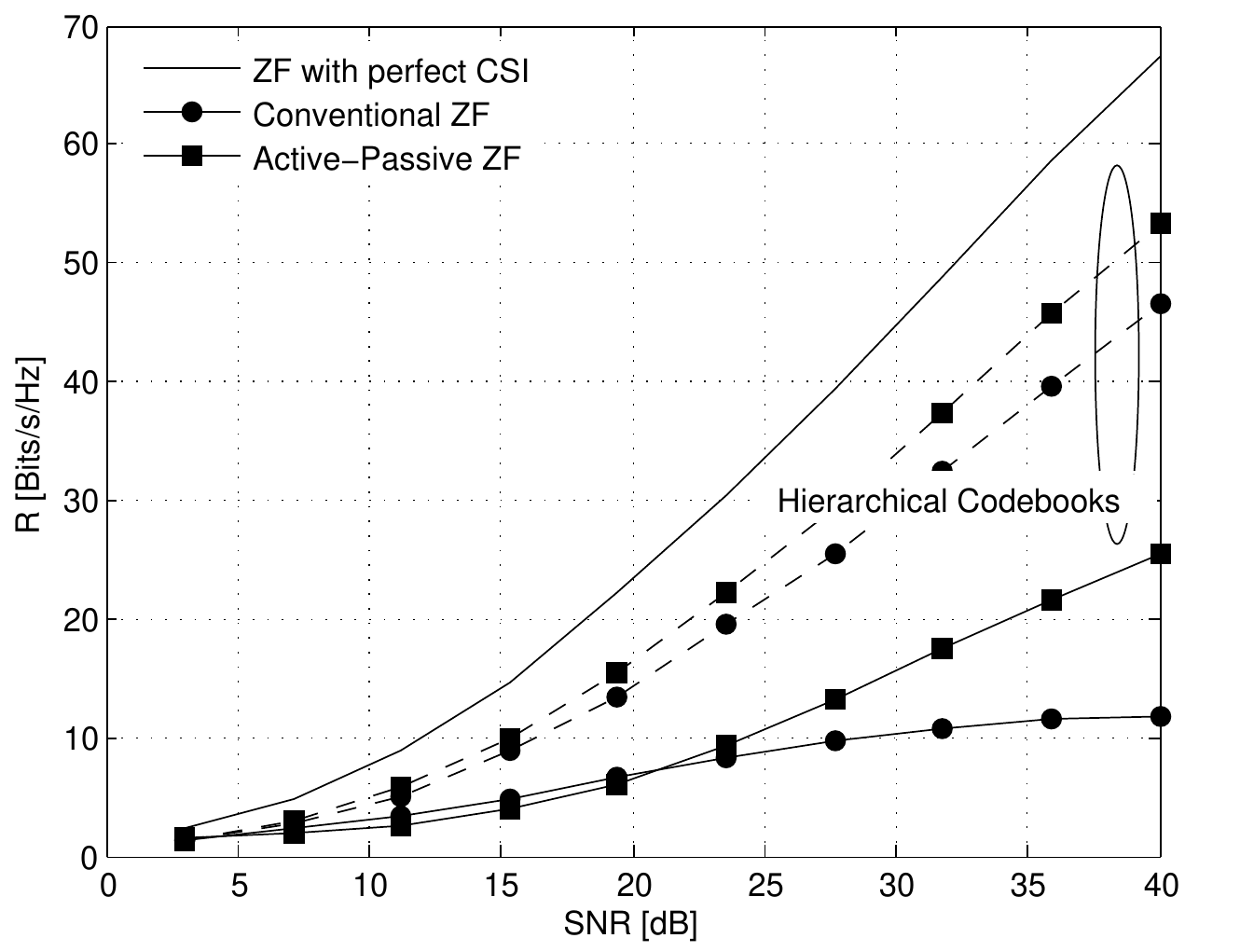}
\caption{Sum rate achieved for the arbitrarily chosen CSI scaling configuration $\alpha$ given in Appendix~\ref{se:App_CSI_matrix}.}
\label{Rate_Chosen_conf_1}
\end{figure} 

In Fig.~\ref{Rate_Chosen_conf_1}, we plot the average sum rate achieved for the previous setting in terms of the SNR. We can observe that the schemes using HQ achieve a much larger number of DoFs (i.e., slope in terms of the SNR) which is in agreement with the theoretical results. Furthermore, the increase in number of DoFs translates to better performance at intermediate SNRs.

\FloatBarrier


\section{Conclusion} 
In this work, we have introduced a new model, called distributed CSI-MIMO channel, consisting in a multicell downlink channel where each transmitter has its own local estimate of the whole multi-user channel. We have shown that conventional ZF precoding applied without taking into account the CSI discrepancies achieves far from the maximal number of DoFs and is limited by the worst accuracy of the CSI over the whole multi-user channel. This is particularly striking as the bad estimate of the channel to one particular user at a unique TX reduces the number of DoFs of all the users. This represents a different behavior from the conventional MIMO BC. In the particular case with only two users, we have provided a precoding scheme achieving the number of DoFs corresponding to the most accurate CSI across the TXs. With arbitrary number of users, the number of DoFs achieved by conventional ZF has been derived and precoding schemes to improve over this number of DoFs value have been provided. Particularly, it has been shown how using codebooks with a hierarchical structure to quantize the CSI could lead to a significant number of DoFs improvement. Moreover, considering the opposite problem of optimizing the sharing of the CSI feedback under a total feedback constraint, we have derived a number of DoFs maximizing CSI configuration when ZF is used. Finally, simulations have confirmed that the novel precoding schemes outperform known linear precoding schemes at intermediate SNRs. 

This paper represents the first step on our work on the DCSI-MIMO channel and many problems remain open. Firstly, the DCSI-MIMO channel has been studied asymptotically for analytical tractability and the extension to finite SNR represents a challenging problem. The design of other robust precoders forms also an interesting problem with a strong potential. Finally, there are many other scenarios where distributed TXs want to cooperate but cannot practically share the exact same CSI (Relay channels, interference channels,...). In such settings, similar analysis could be developed to make the transmission more robust to the CSI discrepancies which are likely to exist in practical settings.

\FloatBarrier

\newpage
\section{Appendix}
\subsection{Some Results on Vector Quantization}\label{se:Appendix_RVQ}
We consider the quantization of the unit-norm complex vector $\tilde{\bm{h}}\in \mathbb{C}^{K}$ over a codebook $\mathcal{C}$ where both the channel to quantize and the elements of the codebook are multiplied by a unit-norm complex number (i.e., are rotated in the complex space) so as to let the first element of the vector be real valued. The quantized vector $\hat{\bm{h}}$ is then obtained as
\begin{equation}
\hat{\bm{h}}=\argmin_{\bm{c}\in\mathcal{C}}\|{\bm{c}}-\tilde{\bm{h}}\|.
\label{eq:App_RVQ_1}
\end{equation}
The rotation is done so as to optimize the performance of the quantization as it clearly leads to better performance. Since the norm is conserved when considering the canonical isomorphism from $\mathbb{C}^{K}$ to $\mathbb{R}^{2K}$, we can consider for the quantization the vectors as elements of $\mathbb{R}^{2K}$ made of the stacked real and imaginary parts of the original vector. 

With the first coefficient real valued,  it is only necessary to consider~$\mathbb{R}^{2K-1}$. Thus, a vector $\bm{u}=[u_1,u_2,\ldots,u_K]^{\trans}\in\mathbb{C}^{K}$ with its first coefficient real valued is represented in $\mathbb{R}^{2K-1}$ as $\bm{u}_{\mathbb{R}^{2K-1}}$ and is defined as
\begin{equation}
\bm{u}_{\mathbb{R}^{2K-1}}\triangleq\begin{bmatrix}
\mathrm{Re}(u_1)&
\mathrm{Re}(u_2)&
\hdots&
\mathrm{Re}(u_K)&
\mathrm{Im}(u_2)&
\mathrm{Im}(u_3)&
\hdots&
\mathrm{Im}(u_K)
\end{bmatrix}^{\trans}.
\label{eq:App_RVQ_2}
\end{equation}
We can then define the angle between $\bm{u}_{\mathbb{R}^{2K-1}}$ and $\bm{v}_{\mathbb{R}^{2K-1}}$ in $\mathbb{R}^{2K-1}$ as 
\begin{equation}
\angle(\bm{u}_{\mathbb{R}^{2K-1}},\bm{v}_{\mathbb{R}^{2K-1}})\triangleq\arccos\left(\frac{\bm{u}_{\mathbb{R}^{2K-1}}^{\trans}\bm{v}_{\mathbb{R}^{2K-1}}}{\norm{\bm{u}_{\mathbb{R}^{2K-1}}}\norm{\bm{v}_{\mathbb{R}^{2K-1}}}}\right).
\label{eq:App_RVQ_3}
\end{equation}
Using the conservation of the norm by the canonical isomorphism, the quantization in~\eqref{eq:App_RVQ_1} is rewritten as
\begin{align} 
\hat{\bm{h}}_{\mathbb{R}^{2K-1}}&=\argmin_{\bm{c}_{\mathbb{R}^{2K-1}}\in\mathcal{C}_{\mathbb{R}^{2K-1}}}\|\bm{c}_{\mathbb{R}^{2K-1}}-\tilde{\bm{h}}_{\mathbb{R}^{2K-1}}\|^2\label{eq:App_RVQ_4_1}\\
&=\argmin_{\bm{c}_{\mathbb{R}^{2K-1}}\in\mathcal{W}_{\mathbb{R}^{2K-1}}}(2-2\bm{c}_{\mathbb{R}^{2K-1}}^ {\trans}\tilde{\bm{h}}_{\mathbb{R}^{2K-1}}).
\label{eq:App_RVQ_4_2}
\end{align}
We can see from~\eqref{eq:App_RVQ_4_2} that the quantization scheme aims at maximizing~$\bm{c}_{\mathbb{R}^{2K-1}}^ {\trans}\tilde{\bm{h}}_{\mathbb{R}^{2K-1}}$. This figure of merit can be linked to the commonly used chordal distance~$\mathrm{d}(\bullet)$ which is defined for two vectors as~\cite{Dai2008}
\begin{align} 
\mathrm{d}(\bm{c}_{\mathbb{R}^{2K-1}},\tilde{\bm{h}}_{\mathbb{R}^{2K-1}})&=\sqrt{\sin^2(\angle(\bm{c}_{\mathbb{R}^{2K-1}},\tilde{\bm{h}}_{\mathbb{R}^{2K-1}}))}\label{eq:App_RVQ_5_1}\\
&=\sqrt{1-|\bm{c}^{\trans}_{\mathbb{R}^{2K-1}}\tilde{\bm{h}}_{\mathbb{R}^{2K-1}}|^2}.
\label{eq:App_RVQ_5_2}
\end{align}
Thus, minimizing the chordal distance is equivalent to maximizing~$|\bm{c}^{\trans}_{\mathbb{R}^{2K-1}}\tilde{\bm{h}}_{\mathbb{R}^{2K-1}}|^2$. This is then equivalent to the quantization scheme~\eqref{eq:App_RVQ_3} if the half-space where~$\tilde{\bm{h}}_{\mathbb{R}^{2K-1}}$ belongs is known. This requires solely one additional bit. Since we are interested in the scaling of the number of bits, this will not make any difference. Consequently, we will study in the following the quantization scheme based on the minimization of the chordal distance
\begin{equation} 
\hat{\bm{h}}_{\mathbb{R}^{2K-1}}=\argmin_{\bm{c}_{\mathbb{R}^{2K-1}}\in\mathcal{C}_{\mathbb{R}^{2K-1}}}\sqrt{\sin^2(\angle(\bm{c}_{\mathbb{R}^{2K-1}},\tilde{\bm{h}}_{\mathbb{R}^{2K-1}}))}.
\label{eq:App_RVQ_6}
\end{equation} 
On that account, we now study the quantization scheme given by \eqref{eq:App_RVQ_6} over the Grassmannian manifold of dimensions $(1,2K-1)$ in the field~$\mathbb{R}$ (i.e., on the unitary ball in~$\mathbb{R}^{2K-1}$). This quantization scheme is studied (in a much more general form) in~\cite{Dai2008} and we start by recalling some results. We then derive some new properties which will be needed in the derivations. \footnote{We will do the abuse of notation consisting in removing the index $\bullet_{\mathbb{R}^{2K-1}}$ in the derivations but it will be clear that any mention of an angle will refer to the angle defined in $\mathbb{R}^{2K-1}$.}

\begin{proposition}[\cite{Dai2008}, Corollary~$2$]
The cumulative distribution function (CDF) of $\mathrm{d}^2(\tilde{\bm{h}},\bm{c})\triangleq\sin^2(\angle(\tilde{\bm{h}},\bm{c}))$ where $\bm{c}\in \mathbb{R}^{2K-1}$ is an element of a random codebook is bounded as	
\begin{equation} 
c_{2K-1}x^{K-1}\leq\mathrm{F}(x)\triangleq\mathrm{Pr}\{\sin^2(\angle(\tilde{\bm{h}},\bm{c}))\leq x\}\leq c_{2K-1}x^{K-1}(1-x)^{\frac{-1}{2}}.
\label{eq:App_RVQ_7}
\end{equation}
where $c_{2K-1}\triangleq \Gamma(K-1/2)/(\Gamma(K)\Gamma(1/2))$.
\label{App_CDF}
\end{proposition}

\begin{proposition}[\cite{Dai2008}, Theorem~$2$]
When the size $L=2^B$ of the random codebook is sufficiently large ($c_{2K-1}^{-1/(K-1)}2^{-B/(K-1)}\leq 1$ is necessary), then it holds that
\begin{equation} 
\frac{2K-1}{2K+1}c_{2K-1}^{-1/(K-1)}2^{-B/(K-1)} \lesssim \E_{\mathcal{C},\tilde{\bm{h}}}[\min_{\bm{c}\in\mathcal{C}}\sin^2(\angle(\tilde{\bm{h}},\bm{c}))]\lesssim \frac{\Gamma(\frac{1}{K-1})}{K-1}c_{2K-1}^{-1/(K-1)}2^{-B/(K-1)}.
\label{eq:App_RVQ_8}
\end{equation} 
\label{App_distorsion}
\end{proposition}

\begin{proposition}
When the size $L=2^B$ of the random codebook is sufficiently large, the expectation of the logarithm of the quantization error is bounded as
\begin{equation} 
\frac{B+\log_2(c_{2K-1})}{(K-1)}\lesssim  \E_{\mathcal{C},\tilde{\bm{h}}}\left[-\log_2\left(\min_{\bm{c}\in\mathcal{C}}\sin^2(\angle(\tilde{\bm{h}},\bm{c}))\right)\right]\lesssim \frac{B+\log_2(c_{2K-1})+\log_2(e)}{(K-1)}.
\label{eq:App_RVQ_9}
\end{equation} 
\label{App_log_distorsion} 
\end{proposition}
\begin{proof} 
\emph{Upper Bound:} The derivation of an upper bound follows the same idea as the proof in Appendix~$B$ of \cite{Dai2008} which derives an upper bound for the same expectation as in this proof, only without the logarithm. We start by recalling a Lemma from \cite{Dai2008} which follows easily from the definition but is helpful.
\begin{lemma}[\cite{Dai2008}, Lemma~$3$]
The empirical distribution function minimizing the distorsion over a given $L=2^B$ is 
\begin{equation}
 \mathrm{F}^{*}_{\mathcal{C}^{*}}(x)=
     \begin{cases}
        0 & \text{if $x<0$} \\
        L \mathrm{F}(x)& \text{if $0\leq x\leq x^{*}$} \\
        1 & \text{if $x>x^{*}$}
     \end{cases}     
\end{equation}
where $x^{*}$ satisfies $L\mathrm{F}(x^{*})=1$ and $\mathrm{F}(x)\triangleq\mathrm{Pr}\{\sin^2(\angle(\tilde{\bm{h}},\bm{c}))|\leq x\}$.
\label{App_lemma_CDF}
\end{lemma}
Note that Lemma~\ref{App_lemma_CDF} corresponds to the optimal codebook minimizing the average distance and is thusly a lower bound for the distorsion. We can then write
\begin{align}
\E_{\mathcal{C},\tilde{\bm{h}}}\left[-\log\left(\min_{\bm{c}\in\mathcal{C}}\sin^2(\angle(\tilde{\bm{h}},\bm{c}))\right)\right]
&=\int_0^{\infty}\mathrm{Pr}\{-\log\left(\min_{\bm{c}\in\mathcal{C}}\sin^2(\angle(\tilde{\bm{h}},\bm{c}))\right)\geq z\}\mathrm{d}z\label{eq:App_RVQ_10_1}\\
&=\int_0^{\infty}\mathrm{Pr}\{\min_{\bm{c}\in\mathcal{C}}\sin^2(\angle(\tilde{\bm{h}},\bm{c}))\leq e^ {-z}\}\mathrm{d}z\label{eq:App_RVQ_10_2}\\
&\leq \int_0^{-\log(x^*)}\!\!\!\!\!\!\!\!\! \mathrm{d}z+ \int_{-\log(x^*)}^{-\infty} L\mathrm{Pr}\{\sin^2(\angle(\tilde{\bm{h}},{\bm{c}}))\leq e^ {-z}\}\mathrm{d}z \label{eq:App_RVQ_10_3}
\end{align}
where \eqref{eq:App_RVQ_10_1} is obtained by exploiting the fact that the term in the expectation is a positive random variable and \eqref{eq:App_RVQ_10_3} follows from the previous lemma since the optimal codebook has a CDF taking larger value that the CDF for a random codebook for every value of the argument $x$.
 
Following the same approach as the proof in Appendix~$B$ of \cite{Dai2008}, we define $\mathrm{F}_0(x)\triangleq c_{2K-1} x^{K-1}$ and $x_0$ so that $L\mathrm{F}_0(x_0)=1$. Let also define $\mathrm{F}_{\mathrm{ub}}(x)\triangleq c_{2K-1} x^{K-1}(1-x)^{-1/2}$ and $x_{\mathrm{ub}}$ so that $L\mathrm{F}_{\mathrm{ub}}(x_{\mathrm{ub}})=1$. Finally, we define $\mathrm{F}_{\mathrm{ubub}}(x)\triangleq c_{2K-1} x^{K-1}(1-x_0)^{-1/2}$ and $x_{\mathrm{ubub}}$ so that $L\mathrm{F}_{\mathrm{ubub}}(x_{\mathrm{ubub}})=1$.

It holds by construction that $x_{\mathrm{ub}}\leq x^{*}\leq x_0$ since we know from Proposition~\ref{App_CDF} that $\mathrm{F}_0(x)\leq \mathrm{F}(x)\leq \mathrm{F}_{\mathrm{ub}}(x)$. Clearly it follows that $(1-x)^{-1/2}\leq(1-x_0)^{-1/2}$ for $x\in [0,x_0]$ so that $\mathrm{F}_{\mathrm{ub}}(x)\leq \mathrm{F}_{\mathrm{ubub}}(x)$ for $x\in [0,x_0]$, which finally implies $x_{\mathrm{ubub}}\leq x_{\mathrm{ub}}$. We can then use these relations to derive an upper bound for~\eqref{eq:App_RVQ_10_3}.
\begin{align}
\E_{\mathcal{C},\tilde{\bm{h}}}\left[-\log\left(\min_{\bm{c}\in\mathcal{C}}\sin^2(\angle(\tilde{\bm{h}},\bm{c}))\right)\right]
&\leq\int_0^{-\log(x^*)} \mathrm{d}z+ \int_{-\log(x^*)}^{-\infty} L\mathrm{F}(e^{-z})\mathrm{d}z\label{eq:App_RVQ_11_1}\\  
& \leq \int_0^{-\log(x_{\mathrm{ubub}})}\mathrm{d}z+\int_{-\log(x_{0})}^{\infty}L\mathrm{F}(e^{-z})\mathrm{d}z\label{eq:App_RVQ_11_2} \\
& \leq \int_0^{-\log(x_{\mathrm{ubub}})}\mathrm{d}z+\int_{-\log(x_{0})}^{\infty}L\mathrm{F}_{\mathrm{ubub}}(e^{-z})\mathrm{d}z.\label{eq:App_RVQ_11_3} 
\end{align}
Equation \eqref{eq:App_RVQ_11_2}  follows from $x_{\mathrm{ubub}}\leq x^*\leq x_0$ and \eqref{eq:App_RVQ_11_3} follows from the fact that~$\mathrm{F}_{\mathrm{ub}}(x)\leq \mathrm{F}_{\mathrm{ubub}}(x)$ for $x\in [0,x_0]$. We now replace $\mathrm{F}_{\mathrm{ubub}}(\bullet)$, $x_{\mathrm{ubub}}$, and~$x_0$ by their expressions to evaluate the integral.
\begin{align}
&\E_{\mathcal{C},\tilde{\bm{h}}}\left[-\log\left(\min_{\bm{c}\in\mathcal{C}}\sin^2(\angle(\tilde{\bm{h}},\bm{c}))\right)\right]\nonumber\\
&\leq -\frac{1}{K-1}\log\left(\frac{(1-x_0)^{1/2}}{Lc_{2K-1}}\right)+\frac{Lc_{2K-1}}{(1-x_0)^{1/2}}\int_{-\log(x_{0})}^{\infty}e^{-z(K-1)}
\mathrm{d}z\label{eq:App_RVQ_12_1}\\
&= -\frac{1}{K-1}\log\left(\frac{(1-(Lc_{2K-1})^{\frac{-1}{K-1}})^{1/2}}{Lc_{2K-1}}\right)+\frac{1}{(1-(Lc_{2K-1})^{\frac{-1}{K-1}})^{1/2}(K-1)}\label{eq:App_RVQ_12_2}\\
&= \frac{1}{K-1}\left(\log\left(Lc_{2K-1}\right)+1\right)+o(1)
\label{eq:App_RVQ_12_3}
\end{align} 
as $L$ increases. Dividing by $\log(2)$ yields the final upper bound.

\emph{Lower Bound:} We start from the lower bound for the CDF given in Proposition~\ref{App_CDF}. It has a form very similar to the CDF for the quantization of a complex vector in the unit-ball in $\mathbb{C}^{K}$ which is usually used for multiple-antenna BC. Hence, we adapt the approach of the proof of Lemma~$3$ by Jindal in \cite{Jindal2006} to the current setting. 

From the lower bound in Proposition~\ref{App_CDF}, we write
\begin{equation}
\mathrm{Pr}\{\min_{\bm{c}\in\mathcal{C}}\left(\sin^2(\angle(\tilde{\bm{h}},\bm{c}))\right)\leq z\}\geq1-(1-c_{2K-1} x^{(K-1)})^{L}.
\end{equation}
A lower bound for the expectation of the logarithm can then be calculated as follows.
\begin{align}
\E_{\mathcal{C},\tilde{\bm{h}}}\left[-\log\left(\min_{\bm{c}\in\mathcal{C}}\sin^2(\angle(\tilde{\bm{h}},\bm{c}))\right)\right]
&=\int_0^{\infty}\mathrm{Pr}\{\min_{\bm{c}\in\mathcal{C}}\sin^2(\angle(\tilde{\bm{h}},\bm{c}))\leq e^ {-z}\}\mathrm{d}z\label{eq:App_RVQ_13_1} \\
&\geq\int_0^{\infty}1-(1-c_{2K-1} e^{-z(K-1)})^{L}\mathrm{d}z\label{eq:App_RVQ_13_2} \\ 
&=\int_0^{\infty}1-\sum_{k=0}^{L}\binom{L}{k}(-1)^{k} c_{2K-1}^k e^{-z(K-1)k}\mathrm{d}z\label{eq:App_RVQ_13_3} \\
&=\frac{1}{K-1}\sum_{k=1}^{L}\binom{L}{k}(-1)^{k+1} \frac{c_{2K-1}^k}{k}\label{eq:App_RVQ_13_4} \\
&=\frac{1}{K-1}f(L)\label{eq:App_RVQ_13_5} 
\end{align}
where we have defined $f(p)\triangleq\sum_{k=1}^{p}\binom{p}{k}(-1)^{k+1} \frac{c_{2K-1}^k}{k}$ for $p\in \mathbb{N}$. To compute the value of $f(L)$, we will use the following relation given in \cite[Sec. $0.155$]{Gradshteyn2007}.
\begin{equation} 
\sum_{k=0}^n\binom{n}{k}\frac{\alpha^{k+1}}{k+1}=\frac{(\alpha+1)^{n+1}-1}{n+1}.
\label{eq:App_RVQ_14}
\end{equation}
We now rewrite $f(L)$ in order to be able to apply \eqref{eq:App_RVQ_14} 
\begin{align}
f(L)
&\triangleq\sum_{k=1}^{L}\binom{L}{k}(-1)^{k+1} \frac{c_{2K-1}^L}{L}\label{eq:App_RVQ_15_1}	 \\
&= (-1)^{L+1} \frac{c_{2K-1}^L}{L}+\sum_{k=1}^{L-1}\left[\binom{L-1}{k-1}+\binom{L-1}{k}\right](-1)^{k+1} \frac{c_{2K-1}^k}{k}\label{eq:App_RVQ_15_2}	 \\
&=\sum_{k=1}^{L}\binom{L-1}{k-1}(-1)^{L+1}\frac{c_{2K-1}^{k}}{k}+\sum_{k=1}^{L-1}\binom{L-1}{k}(-1)^{k+1} \frac{c_{2K-1}^k}{k}\label{eq:App_RVQ_15_3}	 \\
&=\sum_{k'=0}^{L-1}\binom{L-1}{k'}(-1)^{k'+2}\frac{c_{2K-1}^{k'+1}}{k'+1}+f(L-1)\label{eq:App_RVQ_15_4}	 \\
&=-\frac{(-c_{2K-1}+1)^{L}-1}{L}+f(L-1)\label{eq:App_RVQ_15_5}	 \\
&=\sum_{p=1}^{L}\frac{1-(-c_{2K-1}+1)^{p}}{p}\label{eq:App_RVQ_15_6}\\
&=\sum_{p=1}^{L}\frac{1}{p}-\sum_{p=1}^{L}\frac{1-(-c_{2K-1}+1)^{p}}{p}.\label{eq:App_RVQ_15_7}	 
\end{align}
Furthermore we have the two following relations: 
\begin{align}
&\log(L)\leq \sum_{p=1}^{L}\frac{1}{p}\leq \log(L)+1\label{eq:App_RVQ_16_1}\\
&\log(1-x)=-\sum_{L=1}^{\infty}\frac{x^L}{L}\text{ , for $x\in[-1,1]$}.
\label{eq:App_RVQ_16_2} 
\end{align}
Using these properties and dividing by $\log(2)$, we can obtain the final lower bound as 
\begin{align}
\E_{\mathcal{C},\tilde{\bm{h}}}\!\left[-\log\!\left(\min_{\bm{c}\in\mathcal{C}}\sin^2(\angle(\tilde{\bm{h}},\bm{c}))\right)\right]
&\!\geq \!\frac{1}{(K\!-\!1)\log(2)}\sum_{p=1}^{L}\frac{1}{p}-\frac{1}{(K\!-\!1)\log(2)}\sum_{p=1}^{L}\frac{(1\!-\!c_{2K-1})^{p}}{p}\label{eq:App_RVQ_17_1}\\
&\geq \frac{\log_2(L)}{(K-1)}-\frac{1}{(K-1)\log(2)}\sum_{p=1}^{\infty}\frac{(1-c_{2K-1})^{p}}{p}\label{eq:App_RVQ_17_2}\\
&= \frac{\log_2(L)+\log_2(c_{2K-1})}{(K-1)}\label{eq:App_RVQ_17_3}
\end{align} 
where we have used that the constant $c_{2K-1}$ is smaller than one to apply \eqref{eq:App_RVQ_16_2} and obtain the term~$\log_2(c_{2K-1})$.
\end{proof}

\subsection{Proof of Theorem~\ref{thm_DoF_cZF}}\label{se:proof_thm_DoF_cZF}

The proof generalizes to the distributed CSI configuration the proof of Theorem~$4$ in Appendix IV of \cite{Jindal2006}, which derives the number of DoFs for the multiple-antenna BC with finite rate feedback. The generalization is non-trivial due to the fact that in the DCSI-MIMO channel it is not only the inner product between the beamformer and the channel~$\tilde{\bm{h}}_k^{\He}\bm{t}_i^{(j)}$ which matters, but also the coherency between the coefficients used at the different TXs. Following this difference, we do not use the conventional Grassmannian quantization scheme but we use instead the quantization scheme described in Subsection~\ref{se:system_model:CSI}. In a word, it consists in exploiting the fact that the norm is conserved by the canonical isomorphism between $\mathbb{C}^{K}$ and $\mathbb{R}^{2K}$, to use the Grassmannian quantization in the real subspace~$\mathbb{R}^{2K-1}$. The reduction to $2K-1$ real dimensions comes from the multiplication by a unit-norm complex number to let the first coefficient be real valued. We then define the angles between vectors in that real linear space. We refer to Appendix~\ref{se:Appendix_RVQ} for more detail.

The estimation error made at TX~$j$ about the channel vector~$\tilde{\bm{h}}_i$ is denoted by~$\bm{\delta}_i^{(j)}$ such that~$\bm{\delta}_i^{(j)}\triangleq \tilde{\bm{h}}_i-\tilde{\bm{h}}_i^{(j)}$ and the estimation error vectors made at TX~$j$ are stacked in the estimation error matrix~$\bm{\Delta}^{(j)}$ defined as
\begin{equation}
\bm{\Delta}^{(j)}\triangleq 
\begin{bmatrix}
(\bm{\delta}_1^{(j)})^{\He}\\
(\bm{\delta}_2^{(j)})^{\He}\\
\vdots\\
(\bm{\delta}_K^{(j)})^{\He}
\end{bmatrix}.
\label{eq:lemma_1}
\end{equation}
We also denote by $\bm{u}_i^{(j)}\triangleq \bm{t}_i^{(j)}/\norm{\bm{t}_i^{(j)}}$ the (Conventional ZF) unit-norm beamformer computed at TX~$j$ and by $\bm{u}_i^{*}\triangleq \bm{t}_i^{*}/\norm{\bm{t}_i^{*}}$ the same beamformer based on perfect CSI. We omit in this proof the superscript~$\bullet^{c\mathrm{ZF}}$ for clarity. 

Furthermore, we consider in the following that the accuracy of the channel estimates increases with the SNR, i.e., the CSI scaling coefficients~$\alpha_i^{(j)}$ are all positive. If there is one pair of indices~$(i,j)$ for which~$\alpha_i^{(j)}=0$, then the Euclidean distance between~$\bm{u}^{(j)}_k$ and~$\bm{u}^{*}_k$ does not decrease with~$P$ for all~$k$ such that the number of DoFs at all RXs vanishes. When this is not the case, the norm of the channel estimation errors can be approximated as
\begin{align}
\|\bm{\delta}_i^{(j)}\|^2&=\norm{\tilde{\bm{h}}_i^{(j)}-\tilde{\bm{h}}_i}^2\\
&=2-2(\tilde{\bm{h}}_i^{(j)})^{\He}\tilde{\bm{h}}_i\label{eq:lemma_2_1}\\
&=2-2|(\tilde{\bm{h}}_i^{(j)})^{\He}\tilde{\bm{h}}_i|\label{eq:lemma_2_2}\\
&=2-2\sqrt{1-\sin^2(\angle(\tilde{\bm{h}}_i^{(j)},\tilde{\bm{h}}_i))}\label{eq:lemma_2_3}\\
&= \sin^2(\angle(\tilde{\bm{h}}_i^{(j)},\tilde{\bm{h}}_i))+o(\sin^2(\angle(\tilde{\bm{h}}_i^{(j)},\tilde{\bm{h}}_i)))\label{eq:lemma_2_4}
\end{align} 
where \eqref{eq:lemma_2_2} is verified when the channel estimate belongs to the same half-space as the true channel vector. This holds true in this work for the reason explained in Appendix~\ref{se:Appendix_RVQ}. Equality~\eqref{eq:lemma_2_3} follows from the definition of the angle between two vectors and \eqref{eq:lemma_2_4} is obtained via a Taylor expansion on the first order in the estimation error. 

From~\eqref{eq:lemma_1}, we conclude that the square norm of the estimation error~$\|\bm{\delta}_i^{(j)}\|^2$ is asymptotically equal to the chordal distance between the channel estimate and the true channel~$\sin^2(\angle(\bm{h}_i^{(j)},\bm{h}_i))$ when the SNR increases. The chordal distance corresponds to the distance minimized by the Grassmannian quantization so that this will allow us to apply the theoretical results provided in Appendix~\ref{se:Appendix_RVQ}. As a preliminary step, we will now evaluate the impact of the estimation error into the computation of the beamformers, i.e., evaluate the norm of the vector~$\bm{u}_i^{(j)}-\bm{u}_i^*$ for all~$j$. 

\begin{lemma} 
Let's assume that~$\forall i, \alpha_i^{(j)}>0$, then it holds asymptotically as~$P$ increases
\begin{equation}
\E\left[\log_2\left\|\bm{u}_{i}^{(j)}-\bm{u}_{i}^{*}\right\|^2\right]=\E\left[\log_2\left(\max_{i=1,\ldots,K,}\left(\sin^2(\angle(\tilde{\bm{h}}_i^{(j)},\tilde{\bm{h}}_i))\right)\right)\right]+O(1).
\end{equation}
\label{lemma_MSE}
\end{lemma}
\begin{proof}
We consider w.l.o.g. the precoding at TX~$j$. Since $\forall i, \alpha_i^{(j)}>0$, the estimation error is infinitely small as~$P$ increases and we can do a first order approximation of the channel inverse and write
\begin{equation}
\mathbf{H}^{-1}-(\tilde{\mathbf{H}}^{(j)})^{-1}=-\mathbf{H}^{-1}\bm{\Delta}^{(j)}\mathbf{H}^{-1}+o(\norm{\bm{\Delta}^{(j)}}_{\Fro}).
\label{eq:lemma_3}
\end{equation}

\textbf{Derivation of the Upper Bound: } After multiplying by~$\bm{e}_j$ to obtain the $j$-th beamformer, the Right Hand-Side (RHS) of \eqref{eq:lemma_3} can then be upper bounded as follows 
\begin{align}
\|(\mathbf{H}^{-1}-(\tilde{\mathbf{H}}^{(j)})^{-1})\bm{e}_i\|^2
&\leq \|\mathbf{H}^{-1}\bm{\Delta}^{(j)}\mathbf{H}^{-1}\|_{\Fro}^2+o(\norm{\bm{\Delta}^{(j)}}^2_{\Fro})\label{eq:lemma_4_1}\\
&\leq \|\mathbf{H}^{-1}\|_{\Fro}^4\|\bm{\Delta}^{(j)}\|_{\Fro}^2+o(\norm{\bm{\Delta}^{(j)}}^2_{\Fro})\label{eq:lemma_4_2}\\
&\leq K^2\lambda_{\min}^2(\mathbf{H})(\sum_{k=1}^K \|\bm{\delta}^{(j)}_k\|^2)+o(\norm{\bm{\Delta}^{(j)}}^2_{\Fro})\label{eq:lemma_4_3}
\end{align} 
with~$\lambda_{\min}^2(\mathbf{H})$ denoting the smallest eigenvalue of the channel matrix~$\mathbf{H}$. We then take the expectation of the logarithm of this term according to both the channel estimation error and the channel distribution. The term~$\log(\lambda_{\min}^2(\mathbf{H}))$ is shown to be integrable and its expectation is given in \cite{Edelman1989}. The result follows by upper-bounding each of the estimation errors~$\|\bm{\delta}^{(j)}_k\|^2$ by the error which is asymptotically the largest, i.e., the one corresponding to the smallest~$\alpha_i^{(j)}$.

\textbf{Derivation of the Lower Bound: }we start by factorizing the estimation error matrix as follows
\begin{equation}
\bm{\Delta}^{(j)}=\bar{\bm{\Delta}}^{(j)}\diag([\norm{\bm{\delta}_1^{(j)}},\norm{\bm{\delta}_2^{(j)}},\ldots,\norm{\bm{\delta}_K^{(j)}}])
\label{eq:lemma_5}
\end{equation}
with the columns of~$\bar{\bm{\Delta}}^{(j)}$ consequently normalized to be unit-norm. We then assume w.l.o.g. that the asymptotic largest estimation error corresponds to the channel~$\tilde{\bm{h}}_1$ (i.e., the smallest CSI scaling coefficient is~$\alpha_1^{(j)}$). Furthermore, we consider for the sake of exposition that no other channel has the same CSI scaling coefficient. The proof holds similarly if this condition does not hold. We can then write
\begin{equation}
\begin{aligned}
&\E[\log(\|(\mathbf{H}^{-1}-(\tilde{\mathbf{H}}^{(j)})^{-1})\bm{e}_i\|^2)]=\E[\log(\norm{\bm{\delta}_1^{(j)}}^2)]+2\E[\log(\|\mathbf{H}^{-1}\|_{\Fro}^2)]\\
&+\E[\log(\|(\bar{\mathbf{H}}^{-1}\bar{\bm{\Delta}}^{(j)}\diag([1,\norm{\bm{\delta}_2^{(j)}}/\norm{\bm{\delta}_1^{(j)}},\ldots,\norm{\bm{\delta}_K^{(j)}}/\norm{\bm{\delta}_1^{(j)}}])\bar{\mathbf{H}}^{-1}\bm{e}_i)\|^2)]+o(\E[\log(\norm{\bm{\Delta}^{(j)}}^2_{\Fro})])
\end{aligned}
\label{eq:lemma_6}
\end{equation}
where we have defined~$\bar{\mathbf{H}}^{-1}\triangleq \mathbf{H}^{-1}/\norm{\mathbf{H}^{-1}}_{\Fro}$. The absolute value~$|\log(\|\mathbf{H}^{-1}\|_{\Fro}^2)|$ can be upper bounded as in \eqref{eq:lemma_4_3} by~$|\log(K\lambda_{\min}(\mathbf{H}))|$ whose expectation is shown to exist in \cite{Edelman1989}, thus its expectation also exists. Similarly, the absolute value of the last term of the RHS in~\eqref{eq:lemma_6} can be upper-bounded by an integrable function such that it is also integrable and its expectation is then a $O(1)$, i.e., it remains bounded as the SNR~$P$ increases. This concludes the proof.
\end{proof}

\newtheorem*{proof_thm}{Proof of Theorem~\ref{thm_DoF_cZF}}
\begin{proof_thm}
We will now use Lemma~\ref{lemma_MSE} to prove the theorem. We consider for simplicity that the CSI scaling coefficients are all different. The proof easily extends to the configurations with some coefficients equal and this is done solely to simplify the exposition. We assume w.l.o.g. that the TX with the smallest CSI scaling coefficient is TX~$1$.


\textbf{DoF Lower Bound :}
We denote by~$\bm{u}_i\in \mathbb{C}^{K\times 1}$ the beamforming vector\footnote{The vector~$\bm{u}_i$ corresponds to the normalized version of~$\bm{t}_i$. Yet, it is exactly unit-norm only when all the TXs share the \emph{same} channel estimate. It is otherwise impossible for the TXs to jointly normalize the beamformer based on different channel estimates. This does not represent a problem in practice because the power constraint is exactly fulfilled in average over the channel estimation errors.} such that
\begin{equation}
\forall j\in \{1,\ldots,K\},\quad \{\bm{u}_i\}_j=\{\bm{u}^{(j)}_i\}_j=\frac{\{\bm{t}^{(j)}_i\}_j}{P/K}.
\end{equation}

We start from the number of DoFs expression in~\eqref{eq:SM_6} that we rewrite as 
\begin{align}
\DoF_i^{\mathrm{cZF}}\!&=\!1-\lim_{P\rightarrow \infty}\E_{\mathbf{H},\{\mathcal{W}_{i,j}\}}\!\left[\frac{\log_2\left(1+\tfrac{P}{K}\sum_{k\neq i}\norm{\bm{h}_i}^2|\tilde{\bm{h}}_i^{\He}\bm{u}_k|^2\right)}{\log_2(P)}\right]\!\!\label{eq:proof_1_1} \\
&=\! -\lim_{P\rightarrow \infty}\E_{\mathbf{H},\{\mathcal{W}_{i,j}\}}\!\left[\frac{\log_2\left(\sum_{k\neq i}|\tilde{\bm{h}}_i^{\He}\bm{u}_k|^2\right)}{\log_2(P)}\right].\label{eq:proof_1_2} 
\end{align}

To obtain a lower bound for the number of DoFs, we need to derive an upper bound for the leaked interference in~\eqref{eq:proof_1_2}. We define first the selection matrices~$\mathbf{E}_i=\diag(\bm{e}_i)$ and write
\begin{align}
\left|\tilde{\bm{h}}_i^{\He}\bm{u}_k\right|&=\left|\tilde{\bm{h}}_i^{\He}(\bm{u}_k^{*}+\sum_{j=1}^K\mathbf{E}_j(\bm{u}_{k}^{(j)}-\bm{u}_{k}^{*})\right|\label{eq:proof_2_1} \\
&\leq \sum_{j=1}^K\left\|\bm{u}_{k}^{(j)}-\bm{u}_{k}^{*}\right\|\label{eq:proof_2_2} 
\end{align}
which we insert in \eqref{eq:proof_1_2} to obtain 
\begin{align}
\DoF_i^{\mathrm{cZF}}\! &\geq-\lim_{P\rightarrow \infty}\E_{\mathbf{H},\{\mathcal{W}_{i,j}\}}\!\left[\frac{\log_2\left(\sum_{k\neq i}\sum_{j=1}^K\left\|\bm{u}_{k}^{(j)}-\bm{u}_{k}^{*}\right\|^2\right)}{\log_2(P)}\right] \label{eq:proof_3_1} \\
 &\geq-\lim_{P\rightarrow \infty}\E_{\mathbf{H},\{\mathcal{W}_{i,j}\}}\!\left[\frac{\log_2\left(\sum_{k\neq i}K \max_j(\|\bm{u}_{k}^{(j)}-\bm{u}_{k}^{*}\|^2)\right)}{\log_2(P)}\right].\label{eq:proof_3_2} 
\end{align}
From \eqref{eq:lemma_3}, the TX whose computed beamformer exhibits the largest mean square error~$\|\bm{u}_{k}^{(j)}-\bm{u}_{k}^{*}\|^2$ at arbitrarily large SNR~$P$ is the TX to whom the lowest CSI scaling coefficient belongs, which is by assumption TX~$1$. We can then write 
\begin{align}
\DoF_i^{\mathrm{cZF}}\! &\geq\lim_{P\rightarrow \infty}\frac{\E_{\mathbf{H},\{\mathcal{W}_{i,j}\}}\!\left[-\log_2\left(\sum_{k\neq i} \|\bm{u}_{k}^{(1)}-\bm{u}_{k}^{*}\|^2\right)\right]}{\log_2(P)}\label{eq:proof_4_1}  \\
&\geq\lim_{P\rightarrow \infty}\frac{\E_{\mathbf{H},\{\mathcal{W}_{i,j}\}}\left[-\log_2\left(\max_{i=1,\ldots,K,}\left(\sin^2(\angle(\tilde{\bm{h}}_i^{(1)},\tilde{\bm{h}}_i))\right)\right)\right]}{\log_2(P)} \label{eq:proof_4_2} \\
&\geq\lim_{P\rightarrow \infty}\frac{\min_{i=1,\ldots,K} B_i^{(1)}+\log_2(c_{2K-1})+\log_2(e)}{(K-1)\log_2(P)} \label{eq:proof_4_3} \\
&=\min_{i=1,\ldots,K} \alpha_i^{(1)}\label{eq:proof_4_4} 
\end{align}
where \eqref{eq:proof_4_1} is obtained by permuting the expectation and the limit, \eqref{eq:proof_4_2} follows from Lemma~\ref{lemma_MSE} and we have used Proposition~\ref{App_log_distorsion} to obtain inequality \eqref{eq:proof_4_3}. The last equation \eqref{eq:proof_4_4} corresponds to the smallest CSI scaling coefficient and provides the lower bound.

\textbf{DoF Upper Bound: }
We now derive an upper bound for the number of DoFs, which means a lower bound for the interference. We proceed similarly to \eqref{eq:proof_2_1} but this time to obtain a lower bound for the interference remaining after precoding:
\begin{align}
\left|\tilde{\bm{h}}_i^{\He}\bm{u}_k\right|&=\left|\tilde{\bm{h}}_i^{\He}\bm{a}_k\right|\|\sum_{j=1}^K\mathbf{E}_j(\bm{u}_{k}^{(j)}-\bm{u}_{k}^{*})\|\label{eq:proof_5_1} \\
&\geq \left|\tilde{\bm{h}}_i^{\He}\bm{a}_k\right|\|\mathbf{E}_1(\bm{u}_{k}^{(1)}-\bm{u}_{k}^{*})\|\label{eq:proof_5_2} \\
&=\left|\tilde{\bm{h}}_i^{\He}\bm{a}_k\right||\bm{e}_1^{\He}\bm{b}_k^{(1)}|\|\bm{u}_{k}^{(1)}-\bm{u}_{k}^{*}\|\label{eq:proof_5_3} 
\end{align}
where we have defined~$\bm{a}_k\triangleq \left(\sum_{j=1}^K\mathbf{E}_j(\bm{u}_{k}^{(j)}-\bm{u}_{k}^{*})\right)/\norm{\sum_{j=1}^K\mathbf{E}_j(\bm{u}_{k}^{(j)}-\bm{u}_{k}^{*})}$ and $\bm{b}_k^{(1)}\triangleq (\bm{u}_{k}^{(1)}-\bm{u}_{k}^{*})/\|\bm{u}_{k}^{(1)}-\bm{u}_{k}^{*}\|$. The two vectors forming each of the two inner products in \eqref{eq:proof_5_3} are isotropically distributed so that the expectation of their logarithm is finite. Inserting \eqref{eq:proof_5_3} inside the number of DoFs formula in \eqref{eq:proof_1_2}, we can write the lower bound  
\begin{align}
\DoF_i^{\mathrm{cZF}}\!&=\! \frac{\lim_{P\rightarrow \infty}\E_{\mathbf{H},\{\mathcal{W}_{i,j}\}}\!\left[-\log_2\left(\sum_{k\neq i}|\tilde{\bm{h}}_i^{\He}\bm{u}_k|^2\right)\right]}{\log_2(P)} \label{eq:App_two_12_1}\\
&\geq \! \frac{\lim_{P\rightarrow \infty}\E_{\mathbf{H},\{\mathcal{W}_{i,j}\}}\!\left[-\log_2\left(\|\bm{u}_{1}^{(1)}-\bm{u}_{1}^{*}\|^2\right)\right]}{\log_2(P)} \label{eq:App_two_12_2}\\
& \geq \frac{\lim_{P\rightarrow \infty}\E\left[-\log_2\left(\max_{i=1,\ldots,K,}\left(\sin^2(\angle(\tilde{\bm{h}}_i^{(1)},\tilde{\bm{h}}_i))\right)\right)\right]}{\log_2(P)} \label{eq:App_two_12_3} 
\end{align}
with inequality~\eqref{eq:App_two_12_3} obtained from Lemma~\ref{lemma_MSE}. The proof concludes in the same way as the proof of the upper bound after using Proposition~\ref{App_log_distorsion}.
\end{proof_thm}

\subsection{Numerical Values for the Total Feedback Scaling $\gamma$}\label{se:App_DoF_sharing}
\begin{center}
\begin{tabular}{|c|c|c|}
\hline
Number of transmitting TXs:&Saturation of the number of DoFs:&Activation of the $(n+1)$-th TX: :\\
$n$&$n^2(n-1)$&$n^2(n+1)$\\ 
\hline 
1& 0&2 \\
2& 4&12 \\
3& 18&36 \\
4& 48&80 \\
5& 100&150 \\\hline 
\end{tabular}
\end{center}

\subsection{CSI Scaling Matrix Used in the Simulations}\label{se:App_CSI_matrix}

For Fig.~\ref{Rate_Chosen_conf_1}, the CSI scaling matrix arbitrarily chosen is
\begin{equation}
\bm{\alpha}=\begin{bmatrix} 
    0  &  1 &  1 &   1 &  1  &  1&  1\\
    1  &  1 &  1 &   1 &  1  &  1&  1\\
    1  &  1 &  1 &   1 &  1  &  1&  1\\
    1  &  1 &  1 &   1 &  1  &  1&  1\\
    1  &  1 &  1 &   1 &  1  &  0.3&  1\\
    1  &  1 &  1 &   1 &  1  &  1&  1\\
    1  &  1 &  1 &   1 &  1  &  1&  1
    \end{bmatrix}.
\end{equation}
\\
The number of DoFs achieved with the different precoding schemes read as follows:
\\
\begin{center}
\begin{tabular}{|c|c|}
\hline
Precoding Scheme&Number of DoFs\\\hline
Conventional ZF&0\\
Active-Passive ZF&2.1\\
Conventional ZF with HQ&5.3\\
Active-Passive ZF with HQ&6.3\\
\hline 
\end{tabular}
\end{center}
%
\bibliographystyle{IEEEtran}
\bibliography{Literatur}
\end{document}